\documentclass{article}
\usepackage[pdftex]{graphicx}
\usepackage{amsthm}

\newtheorem{thm}{Theorem}[section]
\newtheorem{lemma}{Lemma}
\newtheorem{cor}[thm]{Corollary}
\newtheorem{obs}{Observation}

\newcommand{\floor}[2]{\ensuremath{\left \lfloor {\frac{#1}{#2}} \right
\rfloor}}
\newcommand{\curl}[2]{\ensuremath{\Gamma_{#1, #2}}}
\newcommand{\RH}[1]{\textsl{RH}\ensuremath{(#1)}}

\newcommand{\vomit}[1]{}

\title{A Combinatorial Bound for Beacon-based Routing in Orthogonal Polygons} \author{Thomas C Shermer \\ School of Computing Science \\ Simon
Fraser University \\ Burnaby, BC \ \ V5A 1S6} 

\begin{document}
\maketitle

\begin{abstract}
	\emph{Beacon attraction} is a movement system whereby
	a robot (modeled as a point in 2D) moves in a free space so as to always
	locally minimize its Euclidean distance to 
	an activated beacon (which is also a point).
	This results in the robot moving directly towards the beacon
	when it can, and otherwise sliding along the edge of an obstacle.
	When a robot can reach the activated beacon by this method, we say that
	the beacon \emph{attracts} the robot.
	A \emph{beacon routing} from $p$ to $q$ is a sequence $b_1, b_2, \ldots, b_k$
	of beacons such that activating the beacons in order will attract a robot from
	$p$ to $b_1$ to $b_2 \ldots$ to $b_k$ to $q$, where $q$ is considered to be a
	beacon.
	A \emph{routing set of beacons} is a set $B$ of beacons such that any
	two points $p, q$ in the free space have a beacon routing with the intermediate
	beacons $b_1, b_2, \ldots b_k$ all chosen from $B$.
	Here we address the question of ``how large must such a $B$ be?'' in orthogonal
	polygons, and show that the answer is ``sometimes as large as \floor{n-4}{3},
	but never larger.''
\end{abstract}

\section{Background}
	\emph{Beacon attraction} has come to the attention of the community recently as
	a model of greedy geographical routing
	in dense sensor networks.
	In this application, each node of the network has a location, and each
	communication packet knows the location of its destination.
	Nodes having a packet to deliver forward the packet to their neighbor that is
	the closest (using Euclidean distance) to the packet's
 	destination \cite{Bose01, Karp:2000:GGP:345910.345953}.
	
	In the abstract geometric setting, the destination point is called a beacon,
	and the message is considered to be a point (or robot) that greedily moves
	towards the beacon.
	The robot, under this motion, may or may not reach the beacon---if it
	does reach the beacon, we say that the beacon \emph{attracts} the robot's
	starting point.
	The attraction relation between points has the flavor of a visibility-type
	relation, with the interesting twist that it is asymmetric: if point $p$
	attracts point $q$, then it does not follow that point $q$ attracts $p$.
	In a series of publications, Biro, Gao, Iwerks, Kostitsyna, and Mitchell have 
	studied various visibility-type questions for beacon attraction, such as
	computing attraction (and inverse-attraction) regions for points, computing 
	attraction kernels, guarding, and routing \cite{BiroIwerks, BiroGaoIwerks,
	biro2013beacon}.  In a recent paper, Bae, Shin, and Vigneron studied
	guarding via attraction in orthogonal polygons
	\cite{DBLP:journals/corr/BaeSV15}.
	
	In beacon-based routing, the goal is to route from a source $p$ to a
	destination $q$ through a series of intemediate points $b_1, b_2, \ldots b_k$
	where $b_1$ attracts $q$, $b_2$ attracts $b_1$, $b_3$ attracts $b_2$, etc., and
	finally $q$ attracts $b_k$.  The idea is that we activate the beacons 
	$b_1, b_2, \ldots b_k$ individually in turn, and then activate a beacon at $q$, 
	and we will have attracted $p$ all of the way to $q$.
	In the application setting, this corresponds to using greedy geographical
	routing for each hop in a multi-hop routing for the packet; beacons correspond
	to \emph{landmark} or \emph{backbone} nodes of the network \cite{4215665}.
	Ad-hoc networks (and to some extent, sensor networks) expect to see messages
	from many different $p$'s to many different $q$'s.
	Thus it is natural to ask whether we can find some set $B$ of backbone nodes
	(beacons) such that one can route from \emph{any} $p$ to \emph{any} $q$ using
	only backbone nodes chosen from $B$. 
	
	We'll call such a set $B$ a \emph{routing set of beacons}.
	Biro et al.\cite{BiroGaoIwerks} studied the
	problem of finding minimum-cardinality routing sets of beacons
	in simple polygons.
	They established that it is NP-hard to find such a minimum-cardinality $B$, and
	that such a $B$ can be as large as, but never exceed, $\floor{n-2}{2}$.
	Biro also conjectured \cite{biro2013beacon} that, in orthogonal polygons,
	such a $B$ could be as large as, but never exceed, \floor{n-4}{4}.
	In this paper, we disprove this conjecture, pinning this maximum
	minimum size at \floor{n-4}{3} instead.

 	We organize the remainder of this paper as follows.
 	In Section 2, we define some more terminology and study the decomposition we use.
 	In Section 3, we investigate
 	    the main technical obstacle to using direct induction on the problem,
 	    which we call \emph{trapped paths}.  We also show there how to overcome
 	    this obstacle.
 	In Section 4, we prove the upper bound (over all orthogonal polygons)
 	    on the maximum size of a minimum-sized routing beacon set.
 	In Section 5, we show a construction for arbitrarily large polygons where the
 	minimum size of the routing beacon set for the polygon matches the upper
 	bound.
 	We give concluding remarks in Section 6.

\section{Preliminaries}		
	\subsection{Attraction}
		\begin{figure}[htbp] 
			\begin{center}
			    \includegraphics*[scale=1.25]{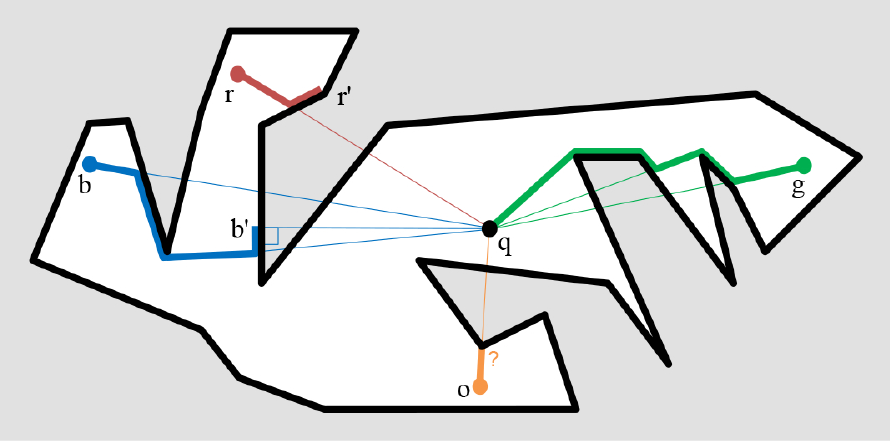} 
			\end{center}
			\caption{
				The movement of points $r, b, g,$ and $o$ under the influence of a beacon at
				$q$. }
			\label{fig:attraction}
		\end{figure}
		We first restrict our attention to polygons.
		Let $p$ be a robot (a mobile point) in a polygon $P$, and $q \in P$ be a
		stationary beacon. We consider the motion of $p$ under the influence of $q$,
		which we call the \emph{attraction path} of $p$ given beacon $q$
		(refer to Figure \ref{fig:attraction}).
		Whenever $p$ can move in a straight
		line towards $q$ inside $P$, then it follows that straight line until it either reaches $q$ or the
		boundary of $P$.
		Whenever $p$ cannot move in a straight line towards $q$ inside $P$, then
		it is on the boundary.
		In this case, it will move along the boundary in the direction that
		decreases its distance to $q$, if such a direction exists.
		The path that $p$ follows may alternate between boundary and
		straight-toward-$q$ sections.
		The figure shows the attraction paths of $r, b, g,$ and $o$ in thick lines,
		with construction lines from $q$ shown in thin lines.
		
		If the attraction path of $p$ given beacon $q$ reaches $q$, then we will
		say that $q$ \emph{attracts} $p$; in the figure, $q$ attracts $g$.
		An attraction path may not reach $q$ for three different reasons.  First,  it
		can become stuck on an edge at a point where the edge is perpendicular to the
		line to $q$, as is the case with $b$ becoming stuck at $b'$ in the figure.
		Second, it can become stuck at a convex vertex with both edges heading away
		from $q$, as is the case with $r$ becoming stuck at $r'$ in the figure.
		Last, a point may start at, or be attracted to, a reflex vertex with both
		edges leading \emph{towards} $q$, as is the case with $o$ in the figure.
		Here the point is not truly stuck, as it may go either direction along the
		boundary.
		In order to resolve the ambiguity here, previous authors have
		adopted a convention that the path always turns to one side or the other
		(say, right) at such reflex vertices \cite{biro2013beacon}.  Here we adopt a
		more conservative approach, saying that the path is \emph{indeterminate} when this happens.
		We will thus be placing our beacons so as to avoid this situation.

		If $q$ attracts $p$, it does not follow that $p$ attracts $q$; for example,
		$g$ does not attract $q$ in the figure.
		This asymmetry of attraction sets it apart from other visibility-type
		relations, which are typically symmetric.
		However, attraction can be placed relative to two known visibility types.
		Firstly, it is a superset of the usual visibility relation: if $p$ and
		$q$ are visible, then $q$ attracts $p$ (and $p$ attracts $q$).
		
		Secondly, in orthogonal polygons (the domain studied here), attraction is a
		subset of the \emph{staircase visibility} relation: if $q$ attracts $p$,
		then $q$ and $p$ are staircase visible.
		(Two points are staircase visible in an orthogonal polygon if there is a path
		$C$ between them in the polygon, composed entirely of horizontal and vertical
		segments, where $C$ is both x-monotone and y-monotone.)
		Staircase visibility is typically not used outside of orthogonal polygons
		and hence the restriction to orthogonal polygons is not onerous.
		
		To see this relation between attraction and staircase visibility, first note
		that attraction paths in orthogonal polygons are always x-monotone and
		y-monotone.  Then consider replacing pieces of the attraction path with
		staircases as suggested in Figure \ref{fig:staircase}---
		the diagonal segments become small-step staircases,
		staying near the attraction segment and therefore in the polygon, and the horizontal and
		vertical segments of the attraction path are left intact in the staircase
		path.
	
		\begin{figure}[htbp] 
			\begin{center}
			    \includegraphics*[scale=1.25]{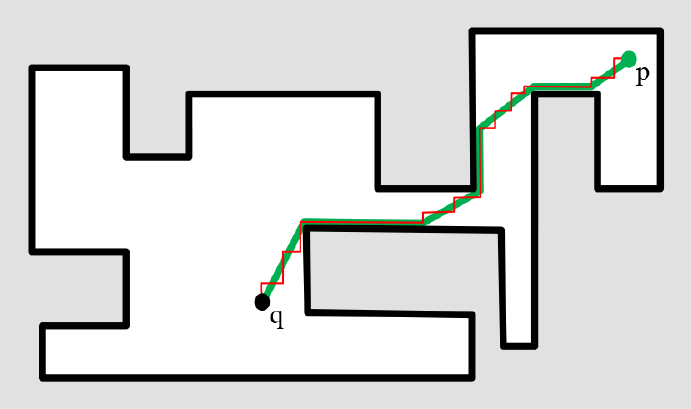} 
			\end{center}
			\caption{
				An attraction path from $p$ to $q$ and a corresponding staircase path.
				}
			\label{fig:staircase}
		\end{figure}
		
	\subsection{Routing segments}
		If $p$ and $q$ are points in a polygon with a beacon routing from $p$ to $q$,
		then by a \emph{routing segment} we mean any maximal section of the
		beacon-routing path during which a point travelling the path is attracted by
		a single beacon (or by the destination point $q$).  If the beacon routing from 
		$p$ to $q$ starts at $p$, proceeds to beacon $b_1$, then to beacon $b_2$, then 
		to $q$, then the routing segments are the part from $p$ to $b_1$, the part
		from $b_1$ to $b_2$, and the part from $b_2$ to $q$.
	
		We will call a routing segment \emph{local} if it is contained in (at most) 
		three rectangles of the decomposition; see Figure \ref{fig:local}. We will
		similarly call a routing path local if all of its segments are local,
		and a routing beacon set local if it supports a local
		routing path between every pair of points in the polygon.
		Our upper bound proof for routing sets of beacons constructs a local routing
		beacon set.		
		
		\begin{figure}[htbp] 
			\begin{center}
			    \includegraphics*[scale=1.25]{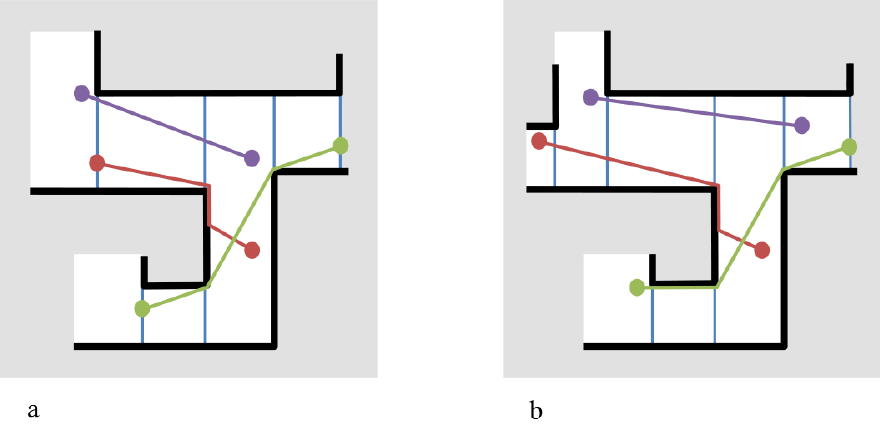} 
			\end{center}
			\caption{
				(a) local paths in the polygon.
				(b) nonlocal paths in the polygon.
			}
			\label{fig:local}
		\end{figure}
	
	\subsection{Decomposition and neighboring rectangles}
		Let $P$ be an orthogonal polygon of $n$ vertices in general position,
		by which we mean that $P$ has no co-vertical or co-horizontal edges.
		One can convert special-position instances to general-position ones
		with the usual perturbation technique, perturbing each edge a symbolic amount
		\emph{into} the polygon.  Moving edges into the polygon avoids creating
		new pairs $p, q$ in the attraction relation.
		
		Construct the \emph{vertical decomposition}
		(also known as the \emph{trapezoidation} \cite{fournier1984triangulating}) of
		$P$ by creating a vertical chord from every reflex vertex (see Figure
		\ref{fig:dualTree}).  We will call these chords the \emph{verticals} of the
		polygon.
		
			Because of our restriction to general position, there are $\frac{n-4}{2}$
		verticals, decomposing the polygon into $\frac{n-2}{2}$ axis-aligned rectangles.
		Each such rectangle has between one and four neighboring rectangles.
		If we form a graph of the neighbor relation on the rectangles, then we have
		the \emph{dual tree} (or \emph{weak dual}) of the decomposition, as shown in
		Figure \ref{fig:dualTree}.
		
		\begin{figure}[htbp] 
			\begin{center}
			    \includegraphics*[scale=1.25]{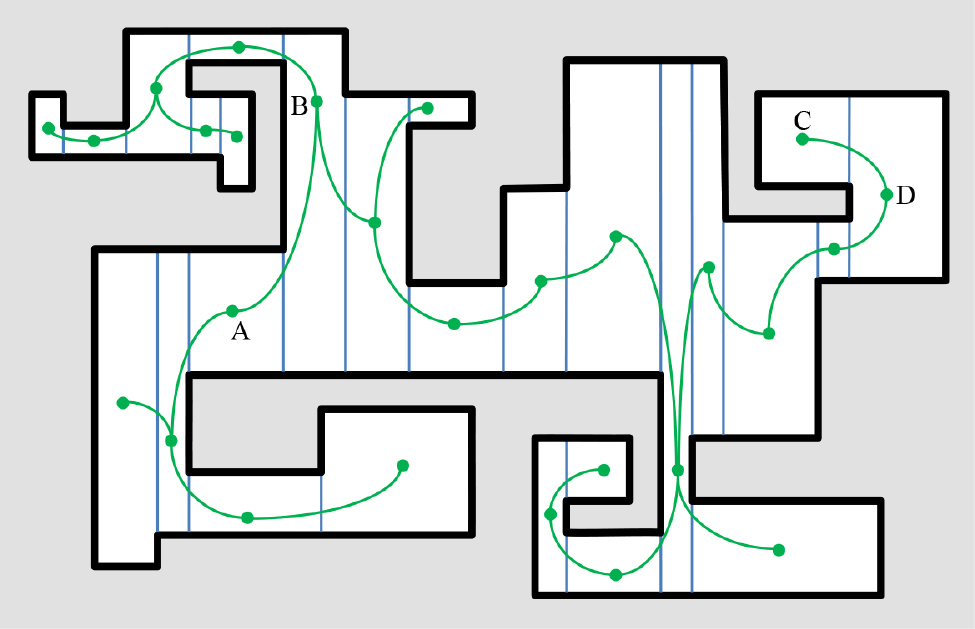} 
			\end{center}
			\caption{The vertical decomposition of a polygon, with its dual tree.}
			\label{fig:dualTree}
		\end{figure}

			We classify the different types of neighbors of a rectangle $R$ in 3 primary
			ways:  \emph{left} vs.\  \emph{right}, depending on the side of $R$ they are
			on; \emph{top} vs.\  \emph{bottom}, depending on whether the neighbor and $R$
			have the same polygon edge along their tops or bottoms; and \emph{short} vs.\ 
			\emph{tall}, depending on whether the neighbor 
			covers a smaller or a larger interval of y-coordinates than $R$ does.
			We combine these classifications: for instance, in Figure \ref{fig:dualTree},
			$A$ is a short bottom left neighbor of $B$, and $D$ is a tall top right neighbor of
			$C$.
		
		\begin{obs}
			If a rectangle $S$ is a is a tall left (or right) neighbor of rectangle $R$,
			then it is the \emph{only} left (or right, respectively) neighbor of $R$.
		\end{obs}
		
		\begin{obs}
			If a rectangle $S$ is a short left (or right) neighbor of rectangle $R$, then
			it is either the \emph{only} left (or right, respectively) neighbor of $R$, 
			or there is one other short left (or right, respectively) neighbor of $R$.
		\end{obs}
		
		If a short neighbor is the only neighbor on a side (left or right) of a
		rectangle, then we call it a \emph{solo} neighbor.  If there is another short
		neighbor on the same side, we call it a \emph{paired} neighbor.
		We generally divide the different cases of a neighboring rectangle's type into
		into tall, solo, and paired.
		Figure \ref{fig:topRight} shows these three types of neighbors.
		 
		\begin{figure}[htbp] 
			\begin{center}
				\includegraphics*[scale=1.25]{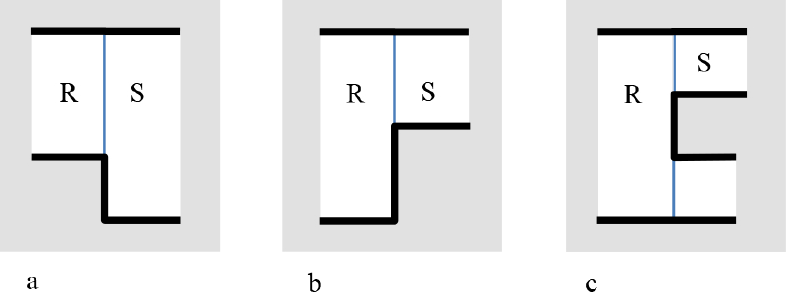} 
			\end{center}
			\caption{
				The three types of top right neighbor $S$ of a rectangle $R$:
				(a) tall, (b) solo, (c) paired. 
			}
			\label{fig:topRight} 
		\end{figure}

	\subsection{Beacon coverage}\label{sec:coverage}
	
		If a point $p$ in a polygon attracts a point $q$, and $q$ attracts $p$, then we
		say that $p$ \emph{covers} $q$.
		Covering amongst points is thus the symmetric subset of the attraction
		relation.
		Using covering allows us to use the same beacon for routing to and from a
		particular point.  If $p$ and $q$ are visible, then $p$ covers $q$, but the
		converse is not necessarily true.

		If $p$ covers every point in some
		region $Q$, then we say that $p$ \emph{covers} $Q$.
		And if there is a set of points $B$ in the polygon such that for every point
		$q$ in region $Q$, there is a $b$ in $B$ that attracts $q$, and a $b'$ in $B$
		that $q$ attracts, then we say that $B$ \emph{covers} $Q$.
		Typically, the point set $B$ will be our set of beacons, and $Q$ will be our
		polygon, or a subpolygon of it.  
		
		Note that this last notion of coverage is \emph{not} ``\ldots there is a $b$
		in $B$ such that $b$ covers $q$''; our notion is more permissive.  We will
		need this permissivity in our proof when we repair trapped paths.
		
		We add the adverb \emph{locally} to
		either type of coverage if that coverage uses only local path
		segments.
		
		To build a routing set of beacons we will mainly use individual beacons to
		cover different regions; the regions are rectangles and their unions.
		So, we start with an investigation of which rectangles of
		the decomposition a beacon covers.
		
		\begin{obs}\label{obs:contained}
		A beacon $b$ locally covers any rectangle of the decomposition it is in.
		\end{obs}
		Note that if $b$ is on a vertical then it will be in two such rectangles.
		
		Let the \emph{rectangular hull} of a pointset $A$, denoted \RH{A}, be the
		smallest axis-aligned rectangle that is a superset of $A$.

		\begin{obs}\label{obs:hullcontained}
		Let $P$ be a polygon containing beacon $b$ and rectangle $R$.
		If \RH{R \cup \{b\}}\  is a subset of $P$,
		then $b$ covers $R$ in $P$. 
		\end{obs}
		
		The lemmas in the remainder of this section establish some beacon placements
		that cover rectangles other than their containing rectangles. 
		
		For the first lemma, we need some definitions. When $S$ is a short neighbor of
		$R$, we call the vertex of $S$ horizontally adjacent to the shared reflex vertex (of
		$R$ and $S$) the \emph{curl vertex of $S$ with respect to $R$}, and denote this
		vertex \curl{S}{R}. 
		(See Figure \ref{fig:soloattraction}c, where $q$ is the curl vertex of $S$ with respect to
		$R$).
		We shorten this phrase if $R$ and/or $S$ is clear or implied.
		
		If a curl vertex of a rectangle is reflex (see Figure \ref{fig:soloattraction}c),
		then it does not necessarily have routing paths similar to other points in
		its neighborhood in $S$.
		Therefore, when dealing with $S$, we will sometimes need to not include the
		curl vertex with it.  We thus define
		\[
			S^* = \left \{ 
					\begin{array}{ll}
						S \setminus \{ \curl{S}{R} \} & 
							\mbox{if \curl{S}{R}\  is a reflex vertex of the polygon } \\
						S & \mbox{otherwise}
					\end{array}
				 \right.
		\]
		
		Finally, if $R$ is a rectangle of the vertical decomposition, and $S$ is a
		side of $R$, then we refer to the intersection of $S$ with the boundary of the polygon as
		a \emph{wall}.
	
	We are now ready to state the first lemma.
	
	\begin{lemma}\label{lem:solocovered}
		If rectangle $S$ is a solo neighbor of rectangle $R$ in the decomposition of
		a polygon, then any point of $R$ locally covers $S^*$, and any point of $S^*$
		locally covers $R$.
	\end{lemma}
	
	\begin{figure}[htbp] 
		\begin{center}
		    \includegraphics*[scale=1.25]{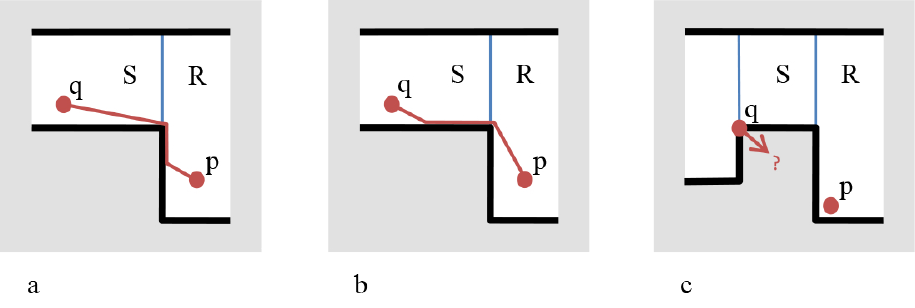} 
		\end{center}
		\caption{
			(a) $p$ is attracted into the left wall of $R$.
			(b) $q$ is attracted into the bottom wall of $S$.
			(c) $q = \curl{S}{R}$ is reflex; some $p$ will give indeterminate results.
		}
		\label{fig:soloattraction}
	\end{figure}
	
	\begin{proof}
		Let $p$ and $q$ be arbitrary points in $R$ and $S$, respectively, and without
		loss of generality, let $S$ be an upper-left neighbor of $R$.
		If $p$ and $q$ are visible, then they mutually attract along their line of
		visibility.
		  
		If $p$ and $q$ are not visible, consider trying to attract $p$ to $q$ by
		activating a beacon at $q$.
		The point will be pulled into the left wall of $R$, and then up along it; once
		it reaches the reflex vertex, it procedes directly to $q$.
		This is illustrated in Figure \ref{fig:soloattraction}a.
		
		Now consider trying to attract $q$ to $p$ by activating the beacon at $p$.
		If $q$ is not the curl vertex, then either it will be pulled into the bottom
		wall of $S$ to the right of the curl vertex (Figure
		\ref{fig:soloattraction}b), or it starts on the bottom wall of $S$ right of
		the curl vertex.  Thereafter it is pulled rightward on that bottom wall until
		it reaches the reflex vertex, where it procedes directly to $p$.
		
		If $q$ is the curl vertex then there is
		the possibility that the vector from $q$ to $p$ points outside of the
		polygon (See Figure \ref{fig:soloattraction}c).  
		Now, $q$ is on one or two edges of the polygon.
		If $q$ is on one edge, it is the edge on the bottom of $S$, and $S$'s left
		neighbor is a bottom neighbor.
		If $q$ is on two edges, forming a convex vertex, then a beacon at $p$
		unambiguously pulls $q$ along the bottom of $S$.
		In either of these cases, the path from $q$ proceeds rightward to
		the reflex vertex and directly to $p$ from there, as was the case with
		all of the other points of $S$.
		
		However, if $q$ is on two edges which form a reflex vertex,
		then the path of attraction is indeterminate; the point could be
		pulled horizontally or vertically.
		In this situation, then, $q$ does not cover $R$.  The lemma follows.
	\end{proof}
	
		We will call a six-sided orthogonal polygon (such as $R \cup S$ in the
		previous lemma) an \emph{L-shaped} polygon.
		Note that the proof above depends only on two edges of the L-shaped polygon
		being polygon boundary: the two edges incident on the reflex vertex.
		
	\begin{lemma}\label{lem:sololeafcovered}
		Let $S$ be a leaf rectangle that is a solo neighbor of rectangle $R$ in the
		decomposition of a polygon $P$, and $b$ be a beacon such that 
		$\RH{R \cup \{b\}} \subset P$.  Then $b$ covers $S$ in $P$.
	\end{lemma}
	\begin{proof}
	The requirement that $S$ is a leaf removes the need for using $S^*$ rather than
	$S$, as leaves do not have reflex curl vertices. Otherwise the situation is the
	same as in the proof of Lemma \ref{lem:solocovered}, with $\RH{R \cup \{b\}}$
	playing the role of $R$ in that proof.  Because the reflex vertex of the
	L-shaped polygon $\RH{R \cup \{b\}} \cup S$ has both incident edges contained
	in the boundary of $P$, that proof applies.
	\end{proof}

	\begin{lemma}\label{lem:tallleafcovered}
		Let $S$ be a leaf rectangle that is a tall neighbor of rectangle $R$ in the
		decomposition of a polygon $P$, and $b$ be a beacon such that 
		$\RH{R \cup \{b\}} \subset P$.  If the two edges of $\RH{R \cup \{b\}} \cup
		S$ incident to its reflex vertex are contained in the boundary of $P$, then
		$b$ covers $S$ in $P$.
	\end{lemma}
	\begin{proof}
		Same as the previous lemma, except that a reflex-incident edge can extend
		past $R$ towards $b$, so the condition on these edges must be made explicit.
	\end{proof}
	
	Next we look at a rectangle with paired neighbors.
	
	Let $R$ have paired neighbors on the left; we define the 
	\emph{left center} of $R$ as the closed
	rectangle that is the full width of $R$ and has the vertical span of 
	the polygon edge on the left of $R$ (as illustrated in Figure
	\ref{fig:pairedattraction}a).
	We furthermore let the \emph{modified left center} of $R$
	be the left center with its two left corners removed.
	
	We similarly define the \emph{right center} and \emph{modified right center} of
	$R$, if $R$ has paired neighbors on the right.
	
	\begin{figure}[htbp] 
		\begin{center}
		    \includegraphics*[scale=1.25]{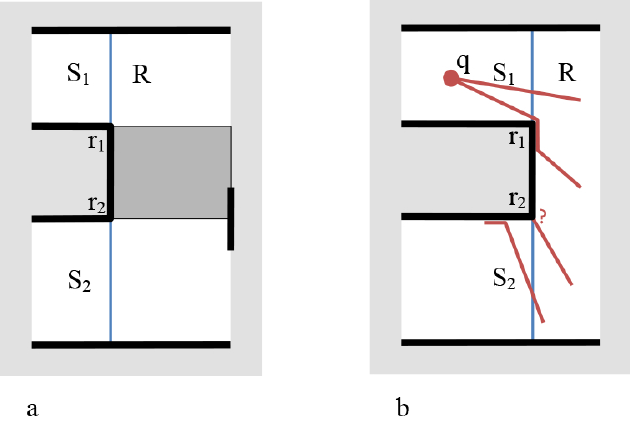} 
		\end{center}
		\caption{
			(a) the left center of $R$ is shown shaded.
			(b) If $p$ is attracted to the left side of $R$ at or above $r_1$, it
			    proceeds into $S_1$ (and directly to $q$). If $p$ is attracted to the
			    left wall of $R$ between $r_2$ and $r_1$, it is pulled up the wall and at
			    $r_1$ will enter $S_1$ and then will reach $q$.
				If $p$ is attracted to the left wall at the point $r_2$, the behavior is
				indeterminate.
				If $p$ is attracted to the left side below $r_2$, it proceeds into $S_2$ and
				does not reach $q$.
		}
		\label{fig:pairedattraction}
	\end{figure}
	
	\begin{lemma}\label{lem:pairedcovered}
		If rectangles $S_1$ and $S_2$ are paired left (or right) neighbors of
		rectangle $R$ in the decomposition, then any point in the modified left (right) center
		of $R$ locally covers $S_1^*$ and $S_2^*$.
	\end{lemma}
	\begin{proof}
		Without loss of generality, let $S_1$ be an upper-left neighbor and $S_2$ be a
		lower-left neighbor of $R$.
		Let $p$ be an arbitrary point in the modified left center of $R$.
		
		By symmetry, we need only show that $p$ covers $S_1^*$.
		Letting $q$ be an arbitrary point in $S_1^*$, we arrive at a situation
		quite similar to that in the proof of Lemma \ref{lem:solocovered}.
		The proof here is the same, except that we need to note that when $p$ is pulled
		towards $q$, if it hits a wall, it hits the wall that is on the left boundary of $R$
		\emph{above the bottom reflex vertex $r_2$}, and therefore is pulled upwards
		(see Figure \ref{fig:pairedattraction}b).  In other words, the last two cases
		of Figure \ref{fig:pairedattraction}b do not occur.
	\end{proof}
	
	We note that $r_2$ is removed from the center as a symmetric counterpart to
	\curl{S_1}{R} in the argument above, and $r_1$ as a counterpart to
	\curl{S_2}{R}.
	
	We will mostly be applying Lemma
	\ref{lem:pairedcovered} with the point in the modified center of $R$ being
	either $r_1 + \varepsilon\hat{x}$ or $r_2 + \varepsilon\hat{x}$.

	\subsection{A small quantity}
	We make use of a small quantity $\varepsilon$, which can be considered
	infinitesimal.
	We could also define it concretely by first taking the the line arrangement
	formed by the lines through every pair of vertices in $P$.
	Then we let $\varepsilon$ be half of the minimum distance between
	intersections of this arrangement.
	
	Let $\hat{x}$ and $\hat{y}$ be unit vectors
	in the x- and y-directions, respectively.
	We will often use $\varepsilon \hat{x}$ or $\varepsilon \hat{y}$ as offsets from
	vertices or other important points in our polygon; Figure \ref{fig:epsilon}
	shows a few of these.
	(In this and in all later figures, the size of $\varepsilon$ is exaggerated.)
	
	\begin{figure}[htbp] 
		\begin{center}
		    \includegraphics*[scale=1.25]{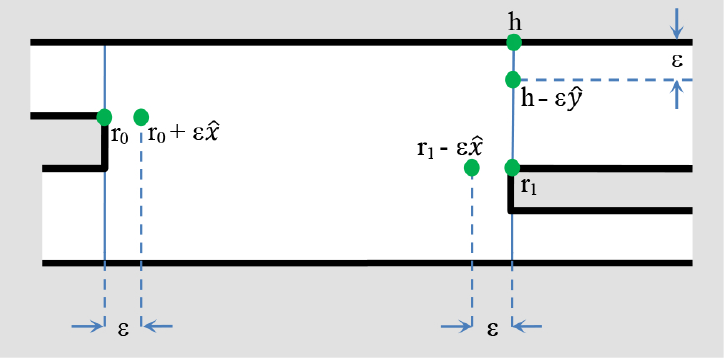}
		\end{center}
		\caption{
			Points $r_0 + \varepsilon \hat{x}$, $r_1 -
			\varepsilon \hat{x}$, and $h - \varepsilon \hat{y}$.
			$\varepsilon$ is not shown to scale; in
			general it would be much smaller.
		}
		\label{fig:epsilon}
	\end{figure}

\subsection{Preparation}
	We will prove the theorem by induction on the size of the dual tree of the
	vertical decomposition.
	We first root the dual tree at an arbitrary leaf.
	At each step, we will examine the structure of the vertical decomposition
	in the vicinity of a deepest node in the rooted tree. 
	We will place some beacons and remove some
	rectangles/dual tree nodes;
	we will place at most two beacons per every three rectangles removed.
	We stop and consider basis cases when the depth of the dual tree reaches
	0, 1, or 2.
	
	We start with a tree $T_0$ that is the entire dual tree
	of the polygon $P$ (which we also denote by $P_0$).
	After step $k$, we will have a tree $T_k$ which is a subgraph of $T_0$,
	with the rectangles corresponding to its vertices forming a 
	single polygon $P_k$ which is a subpolygon of $P$.
	We call each induction step from $T_k$ and $P_k$ to $T_{k+1}$ and $P_{k+1}$ a
	\emph{reduction}.
	
	In a reduction from $P_k$ to $P_{k+1}$, we will let $C_{k+1}$ denote the
	\emph{cut-off region}, which is the closure of $P_k \setminus P_{k+1}$,
	and use $C$ rather than $C_{k+1}$ when the subscript is clear from context.
	Each $C_{i}$ will be the union of some rectangles in the
	decomposition.
	Typically (but not always) $C_{k+1}$ will be connected, and the intersection of
	$C_{k+1}$ and $P_{k+1}$ will then be a vertical $V$.
	In $P_{k+1}$, the vertical $V$ is part of the polygon boundary, but in $P_k$ it
	is not.
	
	If $C_{k+1}$ is not connected, then the intersection of $C_{k+1}$ and $P_{k+1}$
	will be a set of verticals $V, V', \ldots$.  Again, these verticals are part of
	the boundary of $P_{k+1}$ but not of $P_{k}$.

\section{Trapping and repairing paths}\label{sec:repair} 

	To form a beacon set $B_k$ for $P_k$, we would like to take the beacon set
	$B_{k+1}$ for $P_{k+1}$ (which inductively exists) and add a few beacons to it.
	We could use $B_{k+1}$ for routing between pairs of points in $P_{k+1}$ (as a
	subset of $P_k$), and then just worry about routing the points of $C_{k+1}$ (to
	each other, and into and out of $P_k$).  However, this simple strategy does not work,
	because in $P_k$, the beacons $B_{k+1}$ may not be a routing set for the region
	$P_{k+1}$.
	
	This happens because, in rebuilding $P_{k}$ by adding $C_{k+1}$ to $P_{k+1}$,
	the points of $V$ (or $V'$, or $V'', \ldots$) have changed status:
	\begin{itemize}
	\item one end of $V$ changed from a vertex in $P_{k+1}$ to a point in the
		middle of a horizontal edge in $P_k$,
	\item the other end of $V$ changed from a convex vertex or point on a
		vertical edge to a reflex vertex, and
	\item the remainder changed from boundary to non-boundary.
	\end{itemize}
	
	This is important because attraction paths use the boundary in their
	definition.
	
	When $C_{k+1}$ is connected, we will call the rectangle of $C_{k+1}$ containing
	$V$ the \emph{detachment rectangle}, and the rectangle of $P_{k+1}$ containing $V$
	the corresponding \emph{attachment rectangle}.
	When $C_{k+1}$ is not connected, there will be multiple detachment rectangles,
	but they will all have the same attachment rectangle.
	We consider
	the cases of a detachment rectangle $T \subseteq C$ being taller or shorter
	than the corresponding attachment rectangle $R$.
	
	\begin{figure}[htbp] 
		\begin{center}
		    \includegraphics*[scale=1.25]{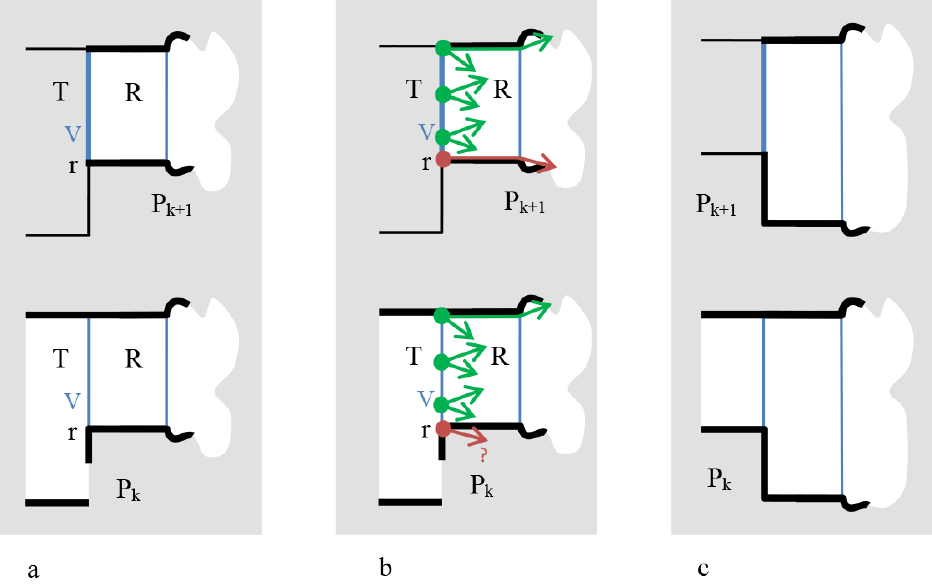} 
		\end{center}
		\caption{
		    (a) $T$ is taller than $R$.
		    (b) Paths to and from $V$ are preserved, except possibly those from $r$.
		    (c) $T$ is shorter than, and a solo neighbor of, $R$.
		}
		\label{fig:repairtallshort}
	\end{figure}
		
	Without loss of generality, we assume that $T$ is an upper-left neighbor of
	$R$.  Consider the case where $T$ is taller than $R$; this is illustrated in
	Figure \ref{fig:repairtallshort}a.
	In $P_{k+1}$, the (relative) interior of $V$ was boundary, but in $P_k$ it is
	not.  We therefore examine all paths in $P_{k+1}$'s routing that are incident
	on $V$.
	 
	In any beacon attraction path, the path can go through the interior of the
	polygon and along some edges.
	Unless a path is entirely collinear with an edge, in order to successfully
	reach the beacon, the only edges along which the path may travel are those
	that have a reflex vertex at the end of the edge it is moving toward.
	Since $V$ neither is in the
	interior of $P_{k+1}$ nor has a reflex vertex on an end in $P_{k+1}$, 
	aside from those paths contained entirely in $V$, no path segments of
	$P_{k+1}$'s beacon routing pass \emph{through} a point of $V$.
	In other words, paths that are incident on $V$ must originate
	or terminate on $V$.
	
	For all points of $V$ other than the bottom  vertex $r$ (reflex in $P_k$),
	these paths that are present in $P_{k+1}$ are also present in $P_k$ (see Figure
	\ref{fig:repairtallshort}b).  For $r$, however, destinations to the right
	and below in $P_{k+1}$ would attract along the horizontal edge, but in $P_k$
	the path cannot choose between the horizontal and vertical edges to start
	($r$ is similar to, but a generalization of, a reflex curl vertex as in Figure \ref{fig:soloattraction}c).
	This problem is easily solved, however, by considering $r$ to be part of $T$
	during the inductive step, obviating the need for it to have
	inductively-generated paths.  The beacon that covers $T$ in the new beacon set
	will also cover $r$.

	Now consider the case where $T$ is shorter than $R$. 
	If $R$ has no other left neighbor, as in Figure \ref{fig:repairtallshort}c, 
	then the edge through $V$ doesn't have a reflex
	vertex at either end in $P_{k+1}$, and thus all paths in $P_{k+1}$ 
	incident on $V$ either originate or terminate there (or both).
	Furthermore, all of these paths are with beacons or points
	lying at or to the right of $V$, so these paths are undisturbed by the
	inductive step.
	
	If $R$ has another left neighbor, then the situation is different.
	The beacons of $P_{k+1}$ may have routings dependent on $V$ being boundary: a
	routing path section may hit the wall of $P_{k+1}$ at a point on $V$ (or start
	on $V$), and then be pulled along that wall until it leaves the wall at some
	reflex vertex (see Figure \ref{fig:trap}a).  In $P_k$, this same section, upon hitting $V$, would
	continue into $T$ and become trapped, not reaching the beacon, as shown in 
	Figure \ref{fig:trap}b.

	\begin{figure}[htbp] 
		\begin{center}
		    \includegraphics*[scale=1.25]{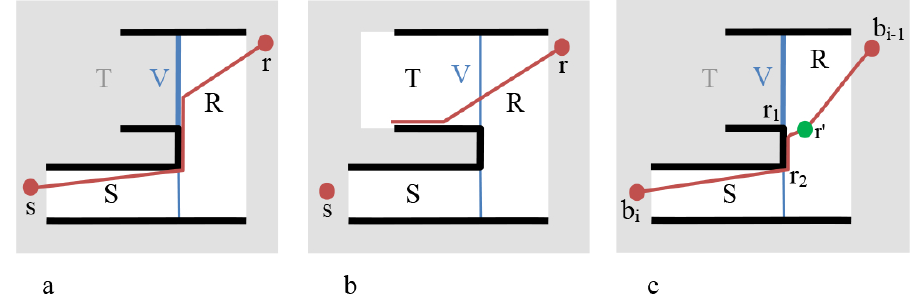} 
		\end{center}
		\caption{
		    (a) a path section hits a wall in $P_{k+1}$.
		    (b) the path continues into $C$ in $P_k$.
		    (c) repairing a section between $b_i$ and $b_{i-1}$ with $r'$.
		}
		\label{fig:trap}
	\end{figure}

	To fix this problem, we will use a new beacon to \emph{repair} such trapped path
	sections, as suggested in Figure \ref{fig:trap}c.
	Let $b_{i-1} b_i$ be a trapped path section of the inductively-generated
	routing beacon set $B_{k+1}$; either or both of the ends of the section may be arbitrary
	points in $P_{k+1}$, and a beacon has been activated at $b_i$.
	By symmetry, without loss of generality assume that the section starts on or
	hits a \emph{left} wall on a rectangle $R$ and is then pulled
	\emph{down} the wall and into another rectangle $S$, as in the figure.
	
	Attraction paths in orthogonal polygons are always both x-monotone and
	y-monotone.
	Thus $b_{i-1}$ is at or right of $V$, and $b_i$ is left of $V$.
	The beacon $b_i$ cannot be colinear with $V$, as then either the path
	would be vertical (and not trapped) or it would hit the left side of $R$
	at $b_i$, not some point on $V$.
	Furthermore, $b_{i-1}$ is either the bottom point $r_1$ of $V$, or above this
	point, and $b_i$ is at or below the reflex vertex $r_2$ on the left of $S$.
	These allowable regions for $b_i$ and $b_{i-1}$ are shown in Figure
	\ref{fig:allowable}.

	\begin{figure}[htbp] 
		\begin{center}
			\includegraphics*[scale=1.25]{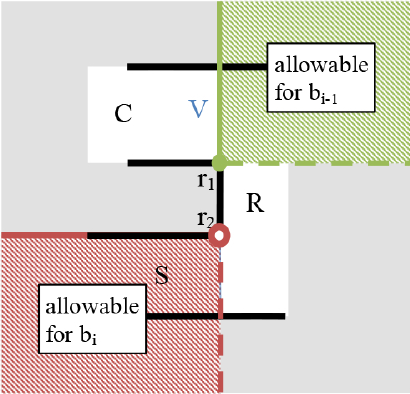}
		\end{center}
		\caption{
		    The allowable regions for $b_i$ and $b_{i-1}$.
		}
		\label{fig:allowable}
	\end{figure}

	Since $R$ has a neighbor $S$ in $P_{k+1}$, and another neighbor $T$
	is connected to it along $V$, which is on the same side as $S$,
	the rectangle $R$ has paired neighbors in $P_k$.
	So we can observe that
	paths can be trapped only when we reduce $P_k$ to $P_{k+1}$ by cutting
	between a rectangle and one of a set of its paired neighbors.
	
	To establish a way to repair trapped paths, 
	we will assume that the inductive routing beacon set is local.  This allows us
	to contain the path section that needs repair in three rectangles: $R, S,$ and
	one other.  This other rectangle is either a left neighbor of $S$ or a right
	neighbor of $R$.

	\begin{lemma}\label{lem:repair}
		Let $B_{k+1}$ be a local routing set of beacons in $P_{k+1}$.
		If a left (or right) paired neighbor $T$ has
		been cut from rectangle $R$ in $P_k$ as part of forming $P_{k+1}$, we can add
		the point $r + \varepsilon\hat{x}$ (or $r - \varepsilon\hat{x}$) to $B_{k+1}$ to obtain a beacon set that supports local
		routing between any pair of points in the subpolygon $P_{k+1}$ of $P_k$, where $r$ is the reflex
		vertex of $P_{k}$ common to $T$ and $R$.
	\end{lemma}

	\begin{proof}
		By symmetry, we need only prove the version where a left paired neighbor is cut
		off.  Let $S$ be the left neighbor of $R$ other than $T$.
		
		Let $r'$ be $r + \varepsilon\hat{x}$ and $B'$ be $B_{k+1} \cup \{ r' \}$.
		Let $b_{i-1}b_i$ be a trapped section of any path of the routing on $B_{k+1}$
		in $P_k$.
		We will replace this section with a pair of (local) sections $b_{i-1}r'$ and
		$r'b_{i}$ when using $B'$.
		We must only establish that these path sections are attractive ($r'$ attracts
		$b_{i-1}$ and $b_i$ attracts $r'$) and local.
		
		The section $b_{i-1}b_i$ in $P_{k+1}$ contains points in the relative interior
		of $S$, as this section proceeds from the left side of $R$ \emph{into} $S$,
		as detailed above in connection with Figure \ref{fig:trap}c.
		It also contains points in the relative interior of $R$, as the points on
		the left side of $R$ above the reflex vertex are relative interior.
		Therefore, being local, $b_{i-1}b_i$ contains points in the relative interior
		of at most one more rectangle.
		We can conclude that $b_{i-1}$ is in $R$ or a right neighbor of $R$,
		and $b_i$ is in $S$ or a left neighbor of $S$ (the only right neighbor of $S$
		is $R$).
		
		Recall that for $b_{i-1}b_i$ to be trapped, $b_{i-1}$ must be $r$ itself, or
		above $r$ (in which case it is also above $r'$).
		Consider what happens when $b_{i-1}$ is attracted by a beacon placed at $r'$
		in $P_k$.  We aim to show that this attraction path is local and reaches
		$r'$.
		
		If $b_{i-1}$ is in $R$, then it is attracted in a straight
		line to $r'$; this is a local section.
		If $b_{i-1}$ is in a right neighbor $A$ of $R$, we consider four cases.
		
		\begin{figure}[p] 
		   \begin{center}
		      \includegraphics*[scale=1.5]{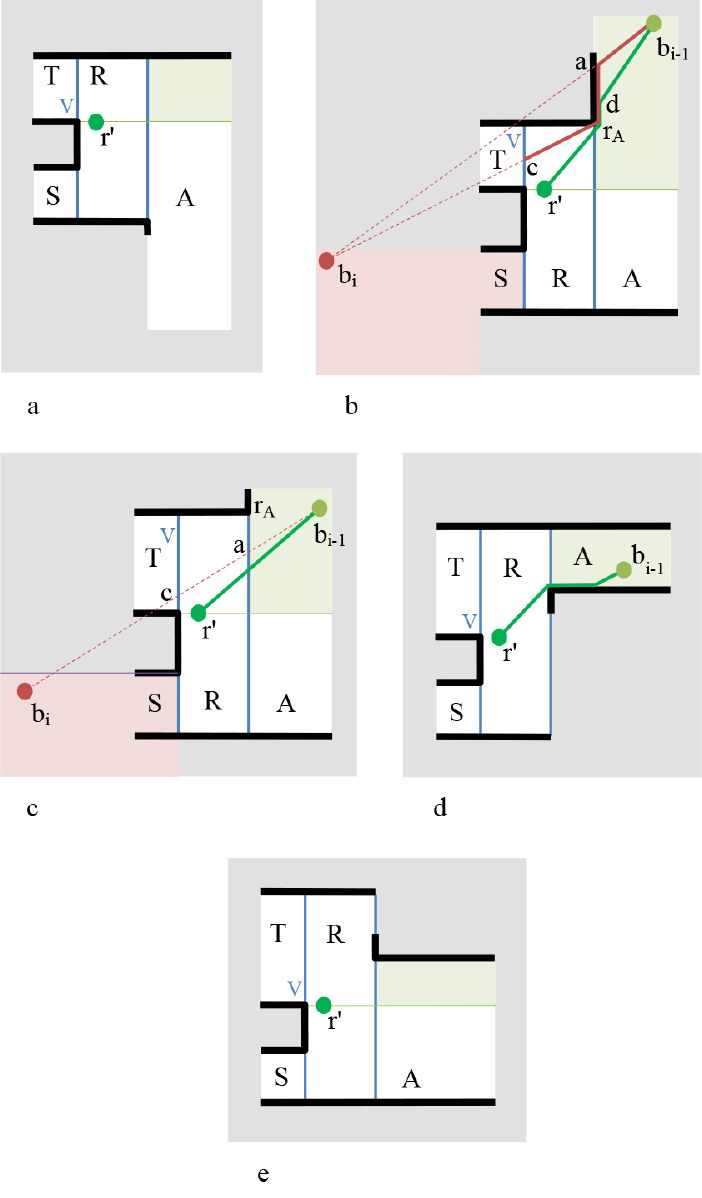} 
		   \end{center}
		   \caption{
		     (a) $A$ is a tall top neighbor of $R$. 
		     (b) and (c) $A$ is a tall bottom neighbor of $R$.
		     (d) $A$ is a short top neighbor of $R$.
		     (e) $A$ is a short bottom neighbor of $R$.
		   }
		   \label{fig:attracttoprime}
		\end{figure}
		
		\begin{description}
		\item[Case A1.  $A$ is a tall top neighbor of $R$.]
			Any $b_{i-1}$ above $r'$ is visible
			(and therefore attracted in a straight line) to $r'$ (see Figure
			\ref{fig:attracttoprime}a).
		
		\item[Case A2. $A$ is a tall bottom neighbor of $R$.]
			Refer to Figure \ref{fig:attracttoprime}b and c.
			Let $r_A$ be the reflex vertex shared between $A$ and $R$.

			$b_i$ attracts $b_{i-1}$ downwards and left,
			owing to the restrictions on their locations (see Figure
			\ref{fig:allowable}).
			Let $a$ be the point where this attraction path first encounters the left
			side of $A$.
			
			If $a$ is above $r_A$, then the situation is as illustrated in Figure
			\ref{fig:attracttoprime}b.
			Here $b_i$ attracts $b_{i-1}$ to $a$, then down the left side of $A$
			to $r_A$, and from there to some point $c$ on $V$.
			$r'$ will attract $b_{i-1}$ to some point $d$ on the left side of $A$.
			If this point is below $r_A$, then $r'$ and $b_{i-1}$ are visible and we
			are done. If $d$ is above $r_A$, then $r'$ will attract the point from $d$
			down the left side of $A$ to $r_A$, and from there directly to $r'$.
			
			The point $d$ is below $a$, because $r'$ must be below $b_{i-1}b_i$.
			This implies that the segment $b_{i-1}d$ is contained in $A$.
			It also implies that $dr_A$ is a subsegment of $ar_A$.
			Since $ar_A$ was boundary in $P_{k+1}$, $dr_A$ was also boundary
			in $P_{k+1}$.
			None of our reductions can trap paths across two verticals, so (with $V$
			being the vertical involved in the trapping here) $ar_A$ and $dr_A$ must also
			be boundary in $P_k$.
			Finally, the segment $r_Ar'$ is contained in $R$, and thus the path
			$b_{i-1}dr_Ar'$ is an attraction path in $P_k$.

			If $a$ is below $r_A$,  then the situation is as in Figure
			\ref{fig:attracttoprime}c.
			Here $b_i$ attracts $b_{i-1}$ directly into $V$ at some point $c$.
			The line segment $cb_{i-1}$ is above $r'$, as $c$ is at or above and
			$b_{i-1}$ is strictly above $r'$.  Thus the line segment $r'b_{i-1}$ is
			below $cb_{i-1}$ and hence contained in $A \cup R$, making it a local path
			segment in $P_{k}$.

		\item[Case A3. $A$ is a short top neighbor of $R$.]
			Refer to Figure \ref{fig:attracttoprime}d. 
			Either $b_{i-1}$ is visible to $r'$
			or $r'$ attracts $b_{i-1}$ into the bottom wall of $A$ at some point $a$;
			this
			path continues left to the reflex vertex $r_a$ shared between $A$ and $R$,
			and then is attracted straight to $r'$.
			Again the path is contained within $A \cup R$ and therefore local.
		
		\item[Case A4. $A$ is a short bottom neighbor of $R$.]
			Refer to Figure \ref{fig:attracttoprime}e. 
			$A$'s top must be above $r'$ in order
			for it to contain the start of a trapped path section.
			In this case, $b_{i-1}$ and $r'$ are visible.
		\end{description}
		
		In each case, we have shown that any $b_{i-1}$ that starts a trapped section
		has a local path section to $r'$.
		
		We now do a similar analysis to show that $r'$ has a local path section to
		any $b_i$ that ends a trapped path section.
		
		As argued above, $b_i$ must be either in $S$ or in a left neighbor $Z$ of $S$.
		If $b_i$ is in $S$, then it attracts $r'$ by Lemma \ref{lem:pairedcovered}.
		Furthermore, the path of this attraction stays within $R \cup S$, so it is
		local.
		
		If $b_i$ is in $Z$, then 
		let $w$ be the lower-left corner of $S$, as in Figure
		\ref{fig:attractfromprime}a.
		
		Consider the relative placement of $r'$ and $b_i$.
		$r'$ is strictly above and to the right of $b_i$
		(recall that $b_i$ must be at the level of, or lower than,
		the reflex vertex $r_S$ common to $R$ and $S$).
		Thus, a beacon at $b_i$ will pull a point at $r'$ along a vector
		that is both downwards and leftwards.
		Since some small neighborhood of $r'$ does not contain any boundary of the
		polygon,
		it is free to travel along that vector, and it thus will not encounter
		polygon boundary until it is strictly below $r'$ (and strictly left of it).
		When it does reach polygon boundary, it is either on the left side of $R$
		between $r_T$ and $r_S$, or on the bottom of $S$ or $R$ (if the vector is
		downwards enough).  We have chosen $\varepsilon$ to be small enough that if
		the line $r'r_S$ hits the line through the bottom of $R$, it hits it 
		either in $S$ or $R$, and not to the left of $S$.
		Equivalently, $\varepsilon$ is small enough that $r'$ is above the line
		$wr_S$.
		
		We now examine two cases, based on where $b_i$ is relative to the line $r'r_S$.
		
		\begin{figure}[htbp] 
			\begin{center}
				\includegraphics*[scale=1.1]{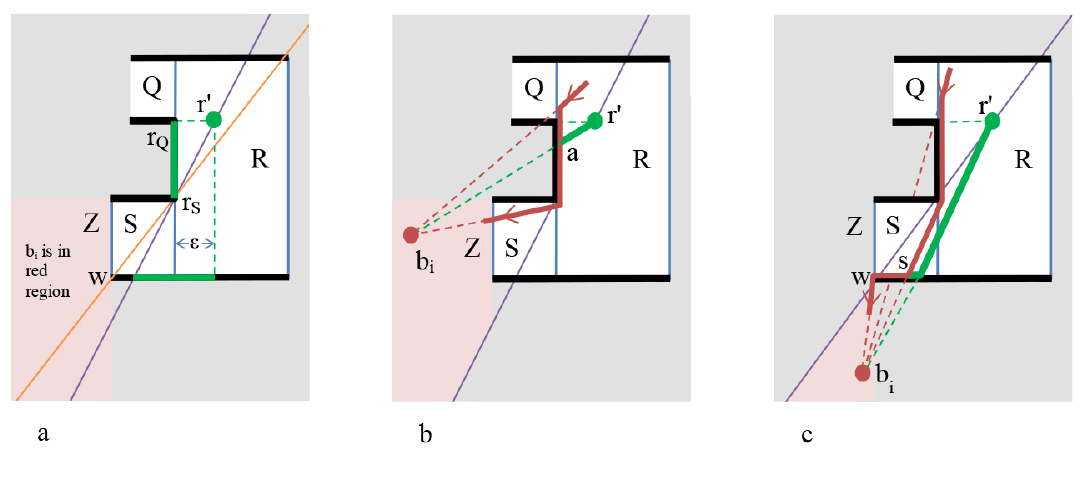} 
			\end{center}
			\caption{
				(a) $r'$ is above the line $wr_s$. 
				(b) $b_i$ attracts $r'$ to $a$.  From then on, it follows $P_{k+1}$'s
					path from $b_{i-1}$ to $b_i$.
				(c) $b_i$ attracts $r'$ to $s$, then leftwards.  It soon encounters $P_{k+1}$'s
					path from $b_{i-1}$ to $b_i$.
			}
			\label{fig:attractfromprime}
		\end{figure}
		
		\begin{description}
		\item[Case Z1. $b_i$ lies on or above the line $r'r_S$.]
			Refer to Figure \ref{fig:attractfromprime}b. 
			In this case, a beacon at $b_i$ attracts $r'$ into a point $a$ on the left side
			of $R$ strictly between $r_1$ and $r_2$.  
			We note that the routing path segment from $b_{i-1}$ to $b_{i}$ in
			$P_{k+1}$ includes the point $a$, as it traverses the entire length of the
			segment $r_1$ to $r_2$.
			Once the point coming from $r'$ hits $a$, it
			will follow the rest of the path from the $b_{i-1}b_i$ section.
			This part of the path is not trapped, being entirely below $V$.
			Thus, there is a valid path segment from $r'$ to $b_i$ in $P_k$,
			and this path segment is local, contained in $Z \cup R \cup S$.
		
		\item[Case Z2. $b_i$ lies below the line $r'r_S$.]
			Refer to Figure \ref{fig:attractfromprime}c.
			In this case, the routing path from $b_{i-1}$ to $b_i$ in $P_{k+1}$,
			after travelling down the left of $R$ to $r_S$, leaves $r_S$ at an angle
			below $r_Sw$ and therefore next encounters the bottom of $S$ at some point $s$.
			It is then pulled leftwards to $w$, which must be a reflex vertex shared by $S$
			and $Z$, and from there it proceeds directly to $b_i$.
			
			A attraction path starting at $r'$ in $P_k$ will either be pulled into
			the bottom of $R$ or $S$.  It is next pulled leftwards to $w$.  At this
			point, or earlier (at $s$), we again start following the old routing path from $b_{i-1}$ to $b_i$, so this
			path also eventually reaches $b_i$.  Again, it is contained in  $Z \cup R \cup
			S$, and is therefore local.
		\end{description}
		
		Now we have shown $r'$ is attracted to any $b_{i}$ that ends a
		trapped section by a local attraction path.  Thus, the two sections
		$b_{i-1}r'$ followed by $r'b_i$ are a valid replacement for any trapped
		section $b_{i-1}b_i$.
	\end{proof} 

	We use the term \emph{repair position} to refer to the placement of the new
	beacon (point) in the previous lemma.
	
	Note that when we repair a path from $b_{i-1}$ to $b_{i}$ by inserting $r'$,
	we do not change the ``reverse'' path from $b_{i}$ to $b_{i-1}$.
	This means that even though our later case analysis will deal only with
	regions \emph{covered} by single beacons, by repair we may end up with regions
	where the symmetry of covering is broken, and routing to a region uses a
	different beacon than routing out of the region does.

\subsection{Routing beacon sets}\label{sec:routingset}
	
	The conditions in the following lemma are sufficient (but not
	necessary) to form a local beacon routing set by inductively cutting off
	a region $C_{k+1}$ from $P_k$ to yield $P_{k+1}$.  Let $A_k(B)$ be the
	attraction relation (digraph) on the points of $B$ in $P_{k}$.

	\begin{lemma}\label{lem:routingset}
		If the following conditions hold, then $B_k = B_{k+1} \cup B'$ is a
		routing beacon set for $P_k$.
		\begin{enumerate}
		  	\item The beacons given ($B'$) locally cover the region $C_{k+1} = P_k
		  	\setminus P_{k+1}$.
		  	\item Each strongly connected component of $A_k(B')$ contains at least
		  		one point in $P_{k+1}$.
		  	\item If a detachment rectangle of $C_{k+1}$ is one of a set of paired
		  	    neighbors of the corresponding attachment rectangle, and the other
		  	    neighbor of the pair is not also a detachment rectangle, then a beacon
		  	    of $B'$ is in repair position.
		\end{enumerate}
	\end{lemma}
	
	\begin{proof}
		The only condition under which inductively-generated paths get trapped is that
		exactly one of a paired set of neighbors of an attachment rectangle is in
		$C_{k+1}$.
		Thus, if there is a possibility of trapped paths, by the third condition we
		have a beacon of $B'$ placed so that we can repair the inductive paths as per
		Lemma \ref{lem:repair}.
		We'll use the term ``repaired induction'' to refer to performing a recursive
		step followed by repair of the paths, if necessary.
		
		If $x$ is a point in $C$, then let $B'(x)$ be a beacon of $B'$ that covers
		$x$.  $B'(x)$ exists by the first condition.  And if $b$ is a beacon in $B'$,
		then let $S(b)$ be a point of $B' \cap P_{k+1}$ that is strongly connected to
		it in $A_k(B')$.  $S(b)$ exists by the second condition.
		
		Consider routing from an
		arbitrary point $p$ to another arbitrary point $q$ in $P_k$.
		Depending on whether each of $p$ and $q$ is in $C_{k+1}$ or not, there are
		four possibilities.

		\begin{description}
		\item[$p$ and $q$ are both in $P_{k+1}$.]
			By repaired induction, there is a local beacon path between $p$ and $q$
			using $B_k$ (plus possibly the beacon in repair position).

		\item[$p$ is in $C_{k+1}$ and $q$ is in $P_{k+1}$.]
			We can route from $p$ directly to $B'(p)$.
			From there, we can route to the beacon $b' = S(B'(p))$ in $B \cap P_{k+1}$,
			by the second condition.
			By the third condition,
			we can then route from $b'$ to $q$ by
			repaired induction.

		\item[$p$ is in $P_{k+1}$ and $q$ is in $C$.] 
			We can ``reverse'' the previous routing,
			routing from $p$ to $S(B'(q))$ by repaired induction,
			from there to $B'(q)$ by the second condition,
			and then directly to $q$.
		
		\item[both $p$ and $q$ are in $C$.]
			We route from $p$ to $B'(p)$, and then to $S(B'(p))$,
			to $S(B'(q))$, to $B'(q)$, and finally to $q$.
		\end{description}
		
		The lemma follows.

	\end{proof}

\section{Reductions}\label{sec:reductions}

	Assume we are after step $k$, having tree $T_k$ and polygon $P_k$ remaining.
	$T_k$ is rooted at a leaf.
	If $T_k$ is of depth 0, 1, or 2, we stop.
	Otherwise, let $L$ be a deepest node in the dual
	tree, let $A_1$ be its direct ancestor (parent),
	and in general let $A_j$ be the direct ancestor of $A_{j-1}$.
	The grandparent $A_2$ of $L$ exists,
	because $T_k$ has depth at least 3.
	In general, we will start by trying to reduce the size of $T_k$
	by removing the dual tree nodes of $A_1$'s subtree;
	this corresponds to cutting the polygon on the vertical chord
	between $A_1$ and $A_2$.  Later we will consider cases that require
	us to examine $A_2$ and its subtree.
			
	We let $A_0$ be synonymous with $L$, and denote the reflex vertex shared between
	$L = A_0$ and $A_1$ as $r_{01}$, and the reflex vertex shared between
	$A_1$ and $A_2$ as $r_{12}$, etc.  
	Other leaves in the vicinity will be denoted $L'$, $L''$, etc.
	and the reflex vertex shared between $L'$ and $A_1$ will be $r'_{01}$, etc.
			
	Throughout this section, all coverage is local and for conciseness we omit the
	adverb, writing \emph{covers} rather than \emph{locally covers}. 
	
	We assume without loss of generality (by symmetry)
	that $A_2$ is an upper right neighbor of $A_1$.  
	With respect to $A_1$, the neighbor $A_2$ is either tall, solo, or paired.
	We first examine the case when $A_2$ is taller than $A_1$.

	\subsection{Case 1: $A_2$ is a tall neighbor of $A_1$}\label{sec:atwotall}
	
		In this case, $A_1$ must have at least one child (the deepest leaf $L$) and can
		have at most two children.  All of $A_1$'s children are left children.

		\begin{lemma}\label{lem:twotalltwokids}
			If $A_2$ is a tall upper right neighbor of $A_1$, 
			and $A_1$ has two children,
			then $P_k$ can be reduced
			by 3 rectangles at a cost of 2 beacons.
		\end{lemma}
		
		\begin{figure}[htbp] 
			\begin{center}
				\includegraphics*[scale=1.25]{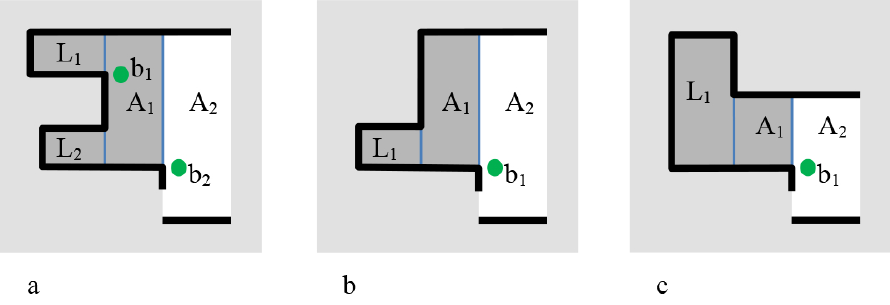} 
			\end{center}
			\caption{
				$A_2$ is a tall neighbor of $A_1$.
				(a) $A_1$ has two children $L_1$ and $L_2$. 
				(b) $A_1$ has a solo lower-left child. 
				(c) $A_1$ has a tall lower-left child.
			}
			\label{fig:twotall}
		\end{figure}
		
		\begin{proof}
			The two children $L_1$ and $L_2$ must be left paired
			children, as shown in Figure \ref{fig:twotall}a.

			This figure also introduces some visual conventions: 
			First, the figure shows the typical local area in $P_k$.
			Second, parts of the boundary of $P_k$ that are \emph{known} to be boundary of
			$P$ are shown with thick black lines.  Parts of the boundary of $P_k$ without
			thick black lines (such as the lower left side of $A_2$ in the figure) may or may not
			be boundary of $P$.
			Third, the beacon placements are shown as green dots.  Beacons placed
			horizontal to and near a reflex vertex (such as both $b_1$ and $b_2$ in the figure) are considered to be placed
			$\pm\varepsilon\hat{x}$ away from them.
			Finally, the choice of which
			rectangles to remove in the reduction are shown as shaded rectangles.

			In this situation,
			we have removed 3 rectangles ($L_1$, $L_2$, and $A_1$) at a cost of placing 2
			beacons ($b_1$ and $b_2$).
			Now we show that, if $P_{k+1}$ has a set $B_{k+1}$ of beacons that allows a
			routing, then $P_k$ has a set of beacons $B_{k} = B_{k+1} \cup \{ b_1, b_2 \}$
			that allows a routing.
 
			Let $C = P_{k} \setminus P_{k+1}$, i.e. $C$ is the union of the
			rectangles $L_1$, $L_2$, and $A_1$.  Also let $B = \{ b_1, b_2 \}$.
			Now the conditions of Lemma \ref{lem:routingset}\ are seen to be satisfied:
			$b_1$ covers the cut-off rectangles $L_1, L_2,$ and $A_1$ (by Lemma
			\ref{lem:pairedcovered}); $b_1$ and $b_2$ are visible, so $B'$ is strongly
			connected in the attraction graph, and $b_2$ is in repair position
			in $P_{k+1}$.
		\end{proof}

		In Case 1, where $A_2$ is taller than $A_1$, it remains to examine the cases
		where $A_1$ has one child.
		We first consider the situation where the one child is a lower neighbor.

		\begin{lemma}\label{lem:lowerleft}
			If $A_2$ is a tall upper right neighbor of $A_1$, and $A_1$ has one
			lower-left child, then $P_k$ can be reduced by 2 rectangles at a cost of 1 beacon.
		\end{lemma}
		\begin{proof}
			The child $L$ is either a short neighbor or a tall neighbor of $A_1$.
			These two cases are shown in  Figure \ref{fig:twotall}b and
			\ref{fig:twotall}c, respectively.  Also shown are the cut-off regions $C = L_1 \cup A_1$ and the
			placement of a beacon $b_1$ to complete the reduction.
			
			By Observation \ref{obs:hullcontained}, $b_1$ covers $A_1$.
			By Lemma \ref{lem:sololeafcovered} or  
			\ref{lem:tallleafcovered}, $b_1$ covers $L_1$.
			$b_1$ is itself (trivially) a strongly-connected graph,
			and it is in $P_{k+1}$.  Furthermore, it is in repair position. 
			Thus by Lemma  Lemma \ref{lem:routingset}, the set of beacons $B_{k+1} \cup
			\{ b_1 \}$ is a routing set.
		\end{proof}

		Now we consider the situation where the one child is an upper-left neighbor.
		We will handle the case of a short upper-left child here, and defer
		the case of a tall upper-left child to Section \ref{sec:deferred}.

		\begin{lemma}\label{lem:upperleft}
			If $A_2$ is a tall upper right neighbor of $A_1$,
			and $A_1$ has one short upper-left child,
			then $P_k$ can be reduced by 2 rectangles at a cost of 1 beacon.
		\end{lemma}
		\begin{proof}
			This situation is shown in  Figure \ref{fig:twotallupper}a,
			along with the rectangles to remove ($C = L_1 \cup A_1$), and the 
			placement of a beacon $b_1$ to complete the reduction.
			
			As in the previous proof, $b_1$ covers $A_1$ and $L_1$.
			$b_1$ is a strongly-connected graph,
			it is in $P_{k+1}$, and is in repair position.
			Thus by Lemma \ref{lem:routingset} the set of beacons $B_{k+1} \cup \{ b_1
			\}$ is a routing set.
		\end{proof}

		Figure \ref{fig:twotallupper}b shows the situation when $L_1$ is a tall
		upper-left child of $A_1$.  This
		situation fails the condition in Lemma \ref{lem:tallleafcovered}.
		Here a beacon at $b_1$ would not suffice, as
		any point of $L_1$ below $b_1$ would not attract $b_1$.

		\begin{figure}[htbp] 
			\begin{center}
				\includegraphics*[scale=1.25]{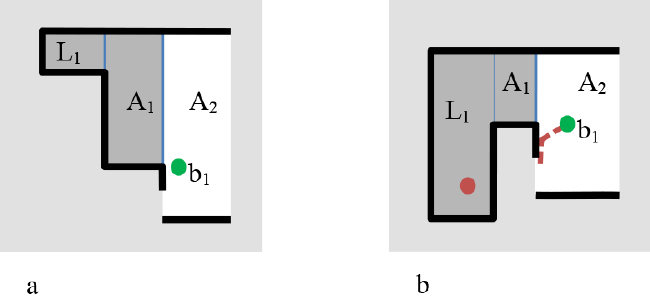} 
			\end{center}
			\caption{ 
				$A_2$ is a tall neighbor of $A_1$.
				(a) $A_1$ has a short upper-left child $L_1$.
				(b) $A_1$ has a tall  upper-left child $L_1$; the point
					$b_1$ is not attracted by the point in $L_1$.
					Here we have exaggerated $\varepsilon$ to make the
					diagram clear.
			}
			\label{fig:twotallupper}
		\end{figure}
		
		The technique we use to handle this case involves examining the structure of
		$A_2$'s subtree.  We defer that analysis until
		Section \ref{sec:deferred}.

	\subsection{Case 2: $A_2$ is a solo neighbor of $A_1$}\label{sec:atwosolo}

		As in the previous case, $A_1$ must have at least one child, can have at most
		two children, and all of its children are left children.  We again start with 
		the case of when $A_1$ has two children.

		\begin{lemma}\label{lem:twosolotwo}
			If $A_2$ is a solo upper right neighbor of $A_1$, and $A_1$ has two children,
			then $P_k$ can be reduced by 3 rectangles at a cost of 2 beacons.
		\end{lemma}
		\begin{proof}
			Because they are both left children, $A_1$'s children must be short children; 
			this situation is shown in Figure
			\ref{fig:twosolo}a, along with the rectangles to remove ($C = L_1 \cup L_2 \cup A_1$), and the 
			placement of beacons $b_1$ at $r_1 + \varepsilon\hat{x}$ and $b_2$ at $q +
			\varepsilon\hat{y}$ to complete the reduction.
			($r_1$ is the reflex vertex shared between $L_1$ and $A_1$, 
			and $q$ is the reflex vertex shared between $A_1$ and $A_2$.)
			
			By Lemma \ref{lem:pairedcovered}, the beacon $b_1$ covers all of $C$, and
			beacon $b_2$ is used only to connect $b_1$ to
			the beacons of $P_{k+1}$.
			The attraction graph on $b_1$ and $b_2$ is strongly-connected,
			as they are visible.
			The beacon $b_2$ is in $P_{k+1}$, and there are no trapped paths to repair.
			Thus, by Lemma \ref{lem:routingset}, the set of beacons $B_{k+1} \cup \{ b_1,
			b_2 \}$ is a routing set.
		\end{proof}
		
		\begin{figure}[htbp] 
			\begin{center}
				\includegraphics*[scale=1.25]{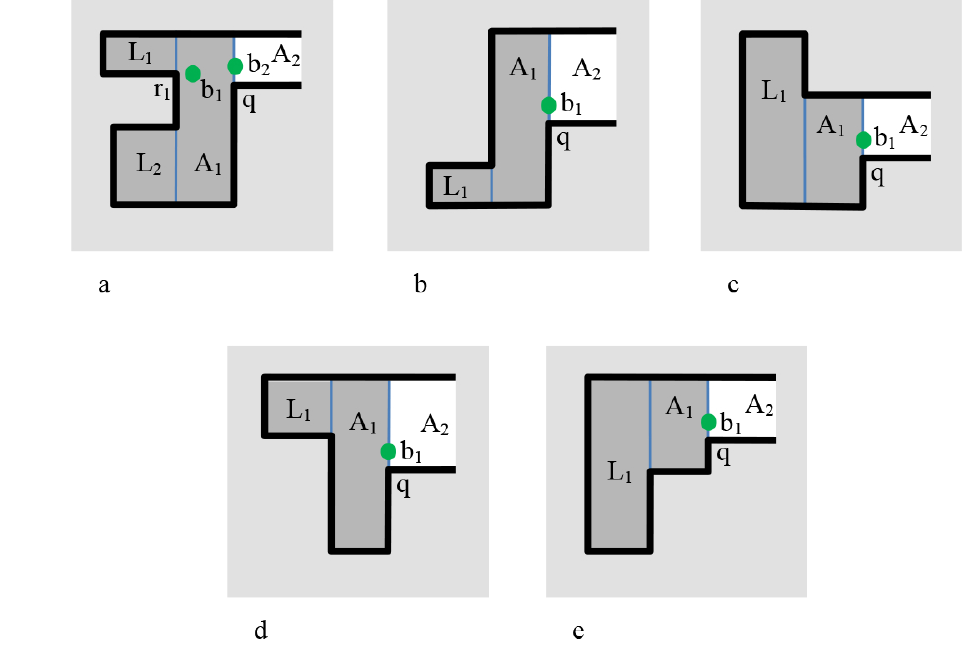} 
			\end{center}
			\caption{ 
				$A_2$ is a solo neighbor of $A_1$.
				(a) $A_1$ has two children $L_1$ and $L_2$.
				(b) $A_1$ has a solo lower-left child. 
				(c) $A_1$ has a tall lower-left child.
				(d) $A_1$ has a solo upper-left child. 
				(e) $A_1$ has a tall upper-left child.
			}
			\label{fig:twosolo}
		\end{figure}

		Since all of the cases when $A_1$ has one child are similar, we handle them in
		one lemma.
		
		\begin{lemma}\label{lem:twosoloone}
			If $A_2$ is a solo upper right neighbor of $A_1$, and $A_1$ has one child,
			then $P_k$ can be reduced by 2 rectangles at a cost of 1 beacon.
		\end{lemma}
		\begin{proof}
			We consider the four possibilities for $A_1$'s child $L_1$:
			either $L_1$ is a lower-left solo neighbor of $A_1$,
			a lower-left tall neighbor,
			an upper-left solo neighbor,
			or an upper-left tall neighbor.
			These possibilities are shown in Figure \ref{fig:twosolo}b--e.
			In each, the beacon $b_1$ is placed at $q + \varepsilon\hat{y}$, where
			$q$ is the reflex vertex shared
			between $A_1$ and $A_2$.
			
			In the various cases, $L_1$ is either a solo neighbor or a tall neighbor of
			$A_1$, and Lemma \ref{lem:sololeafcovered}\ or Lemma
			\ref{lem:tallleafcovered} applies to establish that $b_1$ covers $L_1$.
			By Observation \ref{obs:contained}, $b_1$ also covers $A_1$.
			Since $b_1$ is in $A_2$, and is itself a trivial
			strongly connected graph, the conditions of Lemma \ref{lem:routingset} are
			satisfied and this lemma follows.
		\end{proof} 

	\subsection{Case 3: $A_2$ is a paired neighbor of $A_1$}
			\label{sec:atwopaired} 
		The rectangle paired with $A_2$ as a right neighbor of $A_1$ must be a leaf
		$L_1$.
		$A_1$ must have at least one child; it can have up to three.
		
		\begin{lemma}\label{lem:twopairedthree}
			If $A_2$ is a paired upper right neighbor of $A_1$, and $A_1$ has three
			children, then $P_k$ can be reduced by 4 rectangles at a cost of 2 beacons.
		\end{lemma}
		\begin{proof}
			All of $A_1$'s neighbors must be short, as shown in Figure
			\ref{fig:twopaired}a.
			We place two beacons: $b_1$ at $t + \varepsilon\hat{y}$, where $t$ is
			the lower-left corner of $A_1$, and $b_2$ at $u - \varepsilon\hat{y}$,
			where $u$ is the upper-right corner of $A_1$.
	
			The beacon $b_1$ covers $L_1$ by Observation \ref{obs:hullcontained};
			it also covers $L_3 \cup A_1$, by Observation \ref{obs:contained}.
			The beacon $b_2$ covers $L_2$ by Observation \ref{obs:hullcontained}, and
			it is also a part of $P_{k+1}$.
			$b_1$ and $b_2$ are visible, and thus
			strongly connected in the attraction graph.  
			The reattachment of $A_1$ to $A_2$ causes no paths in $P_{k+1}$ to become
			trapped. By Lemma \ref{lem:routingset}, the result follows.
		\end{proof} 
		
		\begin{figure}[htbp] 
			\begin{center}
				\includegraphics*[scale=1.25]{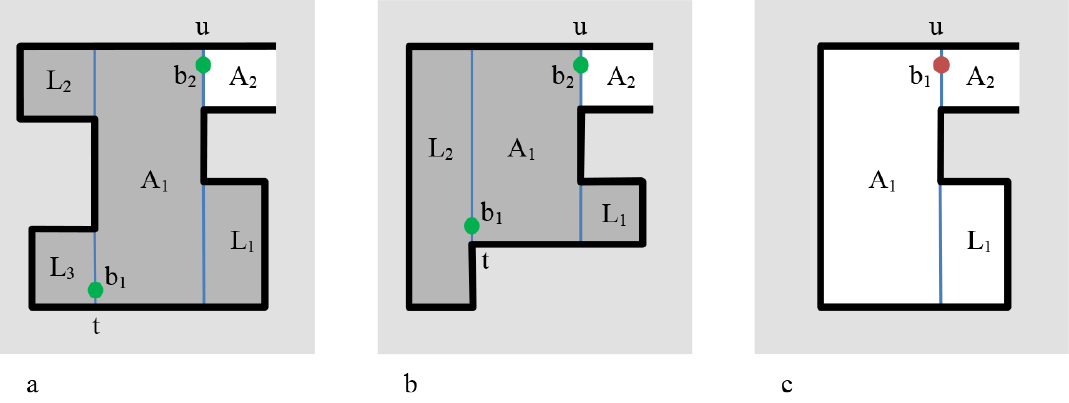} 
			\end{center}
			\caption{ 
				$A_2$ is a paired neighbor of $A_1$.
				(a) $A_1$ has three children. If $A_1$ has two short children,
				$L_2$ or $L_3$ can be removed.
				(b) $A_1$ one short and one tall child. 
				(c) $A_1$ has one short child.
			}
			\label{fig:twopaired}
		\end{figure}
		
		\begin{lemma}\label{lem:twopairedtwo}
			If $A_2$ is a paired upper right neighbor of $A_1$, and $A_1$ has two
			children, then $P_k$ can be reduced by 3 rectangles at a cost of 2 beacons.
		\end{lemma}
		\begin{proof}
			If $A_1$ has two short children, then the situation must be
			as shown in Figure \ref{fig:twopaired}a, with either $L_2$ or $L_3$
			removed.  We can use the same beacon placement and proof as in Lemma
			\ref{lem:twopairedthree}, but remove only 3 rectangles instead of 4.
			
			If one of $A_1$'s children is tall, then the situation is as shown
			in Figure \ref{fig:twopaired}b if it is a tall top neighbor, or a similar
			situation if it is a tall bottom neighbor.  In either case, we 
			place $b_1$ at $t + \varepsilon\hat{y}$, where $t$ is
			the lower-left corner of $A_1$, and $b_2$ at $u - \varepsilon\hat{y}$,
			where $u$ is the upper-right corner of $A_1$.
	
			The beacon $b_1$ covers $C = L_1 \cup L_2 \cup A_1$, and 
			beacon $b_2$ is used only to connect $b_1$ to the beacons of $P_{k+1}$.
			The beacons $b_1$ and $b_2$ are visible, and thus strongly connected in the
			attraction graph.  
			The reattachment of $A_1$ to $A_2$ causes no paths in $P_{k+1}$ to become
			trapped. Thus, by Lemma \ref{lem:routingset}, the result follows.
		\end{proof} 
		
		If $A_1$ has only one child, then it must be the short child $L_1$ that is
		paired with $A_2$, as shown in Figure \ref{fig:twopaired}c.  We would like
		to use a beacon in the same place as $b_2$ in the two- and three-child cases,
		but this beacon ($b_1$ in the figure) is not attracted by all of the points of
		$L_1$.  Thus we must do something different.
		
		We will handle this case by examining one level farther up the dual tree,
		considering $A_2$'s children.  This, along with handling our previously
		deferred case, is done in the next section.

	\subsection{Three-level reduction}\label{sec:deferred}

		Consider $A_2$, the grandparent of some deepest $L$ node in the dual tree.
		If any of the cases handled by Lemmas \ref{lem:twotalltwokids} through
		\ref{lem:twopairedtwo} is present on any of its children, then perform the
		corresponding reduction.
		If one cannot do this, then every height-two subtree of $A_2$ can be pictured
		like the rectangles $A_1 \cup L_1$ in
		either Figure \ref{fig:threeleveltypes}a 
		(deferred from Section \ref{sec:atwotall})
		or Figure \ref{fig:threeleveltypes}b
		(deferred from Section \ref{sec:atwopaired}).
		It is also possible that $A_2$ has one or two height-one subtrees, as shown in
		Figure \ref{fig:threeleveltypes}c.
				
		\begin{figure}[htbp] 
			\begin{center}
				\includegraphics*[scale=1.25]{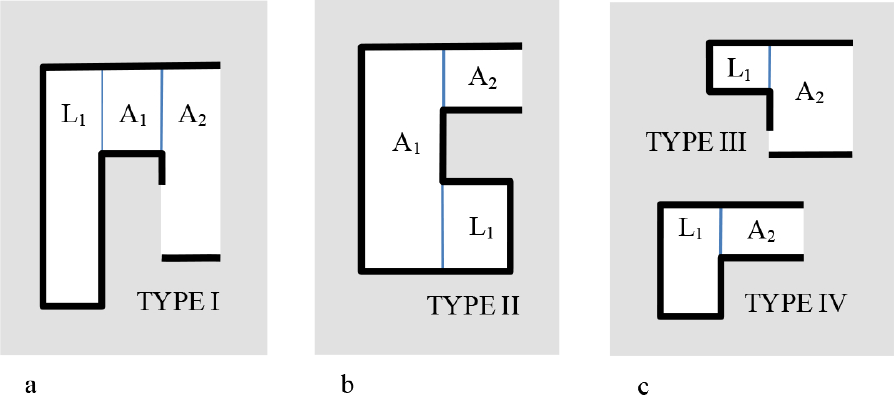} 
			\end{center}
			\caption{ 
				(a) Type I subtree of $A_2$.
				(b) Type II subtree of $A_2$. 
				(c) Types III and IV subtrees of $A_2$.
			}
			\label{fig:threeleveltypes}
		\end{figure}

		We call these subtrees Types I to IV.  In Type I, $L_1$ and $A_2$ are tall
		neighbors of $A_1$, and they all share a horizontal edge.
		In Type II, $L_1$ and $A_2$ are paired neighbors of $A_1$.
		In Type III, the subtree of $A_2$ is a short leaf.
		In Type IV, the subtree of $A_2$ is a tall leaf.
		
		We call Type II and Type IV subtrees \emph{tall}, because they include a tall
		neighbor of $A_2$.  We similarly call Type I and Type III subtrees
		\emph{short}.  Note that \emph{tall} and \emph{short} do not refer to the
		depth of the subtrees (Types I and II have depth 2, and Types III and IV have
		depth 1).
		
		We now remove our assumption that $A_2$ is an upper-right neighbor of $A_1$ in
		order to assume (without loss of generality) that $A_3$ is an upper-right
		neighbor of $A_2$.
		Note that $A_3$ does exist because we have assumed that the depth of the dual
		tree is at least 3.
		
		\begin{lemma}\label{lem:typeII}
			If $A_2$ has a Type II subtree, then $P_k$ can be reduced by 3 
			rectangles at a cost of 2 beacons or 5 rectangles at
			a cost of 3 beacons.
		\end{lemma}
		\begin{proof}
			If $A_2$ has a Type II subtree, then that subtree is on the left of $A_2$,
			because it is tall and $A_3$ is on the right.  Also because it is tall, there
			is no other Type II or Type IV subtree present.  We consider cases based on 
			what the lower right neighbor of $A_2$ is: it can be a Type I subtree, a Type
			III subtree, or it can be absent.  If it is absent, then we further break
			down the situation based on whether $A_3$ is taller than or shorter than
			$A_2$.
			
			\begin{description}
			\item[Case 1: $A_2$ has no lower left neighbor and $A_3$ is shorter
			than $A_2$.] 
			
				This situation is as depicted in Figure \ref{fig:typetwo}a. 
				The rectangle $A_1$ of the Type II subtree is either a tall bottom left
				neighbor of $A_2$ (as depicted) or a tall top left neighbor of $A_2$
				(as suggested by the dashed lines).
				In any case, we put a beacon $b_1$ at $r_{12} - \varepsilon\hat{x}$ and a
				beacon $b_2$ at $u - \varepsilon\hat{y}$, where $r_{12}$ is the reflex vertex
				shared by $A_1$ and $A_2$, and $u$ is the upper-right corner of $A_2$.
				(The red point $b'_1$ on the figure shows where $b_1$ would be if $A_1$
				extends below $A_2$.)
				
				Here we have no trapped paths to repair, because $A_2$
				is taller than $A_3$. 
				The beacon $b_1$ covers $L \cup A_1 \cup A_2$, by Lemma
				\ref{lem:pairedcovered}.
				The two beacons are visible, thus strongly connected, and $b_2$ is in
				$P_{k+1}$.
			
				\begin{figure}[htbp] 
					\begin{center}
						\includegraphics*[scale=1.25]{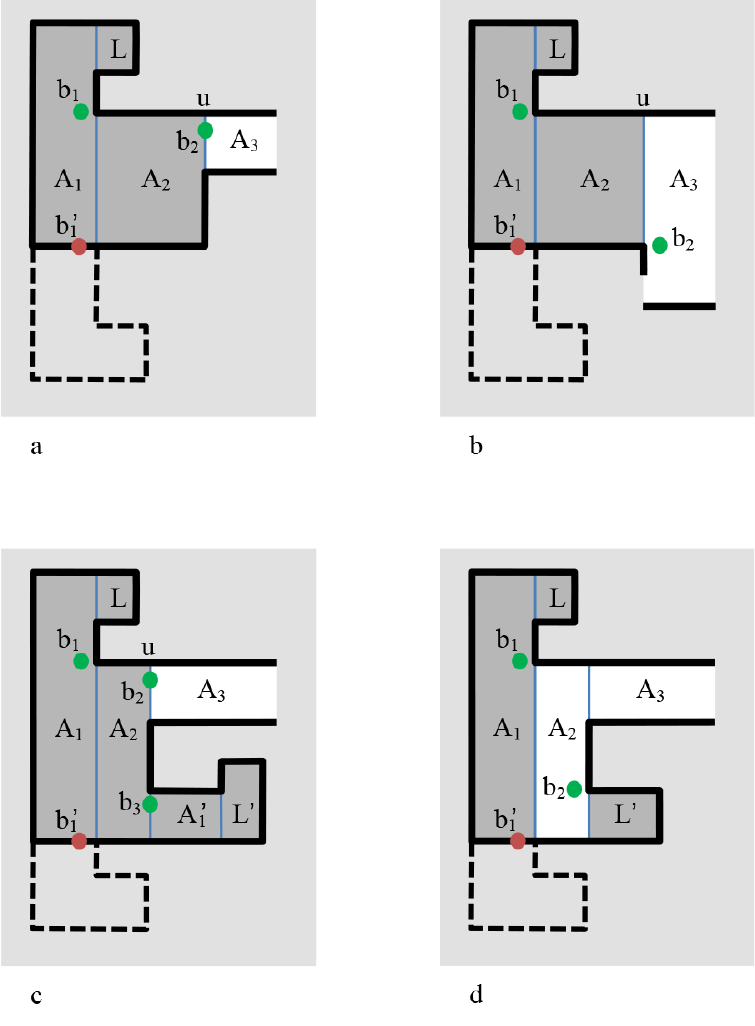} 
					\end{center}
					\caption{ $A_2$ has a Type II left neighbor.
						(a) $A_3$ is a solo neighbor of $A_2$.
						(b) $A_3$ is a tall neighbor of $A_2$.
						(c) $A_3$ is paired with a Type I subtree.
						(d) $A_3$ is paired with a Type III subtree.
					}
					\label{fig:typetwo}
				\end{figure}

			\item[Case 2: $A_2$ has no lower left neighbor and $A_3$ is taller than
			$A2$.] 
			
				This situation is as depicted in Figure \ref{fig:typetwo}b. 
				Again, $A_1$ is either a top or a bottom neighbor, and we depict the first
				while only suggesting the second.
				We put a beacon $b_1$ at $r_{12} - \varepsilon\hat{x}$ and a
				beacon $b_2$ at $r_{23} + \varepsilon\hat{x}$, where $r_{12}$ is the reflex
				vertex shared by $A_1$ and $A_2$, and $r_{23}$ is the reflex
				vertex shared by $A_2$ and $A_3$.
				
				Here we may have trapped paths going through $A_3$, but we have placed $b_2$
				in repair position.  The beacons cover the removed rectangles $A_2, A_1,$ and
				$L$, and by Observation \ref{obs:hullcontained}, the beacons see one
				another.
			
			\item[Case 3: The lower left neighbor of $A_2$ is Type I.] 
					
				Let $L'$ be the leaf of the lower left subtree, and $A'_1$ be its other
				rectangle.  This situation is as depicted in Figure \ref{fig:typetwo}c.
				Here we put a beacon $b_1$ at $r_{12} - \varepsilon\hat{x}$, a beacon $b_2$
				at $u - \varepsilon\hat{y}$, and a beacon $b_3$ at $r'_{12}$, where $b_1$ is
				the reflex vertex shared by $A_1$ and $A_2$, $u$ is the upper-right corner
				of $A_2$, and $r'_{12}$ is the reflex vertex shared by $A'_1$ and $A_2$.
				
				Here $b_1$ covers $A_1$ and $L$, $b_3$ covers $A'_1$ and $L'$, and $b_2$
				covers $A_2$.  The beacons are all visible to one another, so they are
				strongly connected in the attraction graph.  There are no trapped paths, as
				$A_3$ is shorter than $A_2$.  We remove the five rectangles $L, A_1, L', A'_1,$ and 
				$A_2$.

			\item[Case 4: The lower left neighbor of $A_2$ is Type III.] 	
			
				Let $L'$ be the leaf rectangle that is the sole rectangle in the Type III
				subtree to the lower left of $A_2$.  The situation is as depicted in Figure
				\ref{fig:typetwo}d.
				Here we use the technique of removing \emph{two} subtrees of the dual tree:
				the Type II subtree $L \cup A_1$ and the Type III subtree $L'$.
				
				We place beacons $b_1$ at $r_{12} - \varepsilon\hat{x}$, and $b_2$ at
				$r'_{2} - \varepsilon\hat{x}$.  These beacons are visible to one another and
				$b_2$ is in $P_{k+1}$.  Detaching $A_1$ from $A_2$, and then reattaching it,
				cannot create trapped paths because $A_1$ is taller than $A_2$.
				On the other hand, detaching and reattaching $L'$ from $A_2$ will cause some
				paths to become trapped, as $L'$ is a short paired neighbor of $A_2$.
				However, we have placed $b_2$ in repair position for this eventuality.
				The beacon $b_1$ covers $L$ and $A_1$, and $b_2$ covers $L'$.
			\end{description}
			
			In every case, we have reduced the polygon by either 3 rectangles at a cost
			of 2 beacons, or 5 rectangles at a cost of 3 beacons.
			
		\end{proof} 
		
		Next we handle the case when $A_2$ has a Type IV subtree.  In what follows,
		we will use the phrase ``on the vertical'' to mean ``on the relative interior
		of the vertical''--i.e. we do not include the vertical's endpoints as
		allowable positions.
		
		\begin{lemma}\label{lem:typeIV}
			If $A_2$ has a Type IV subtree, then $P_k$ can be reduced by 4 
			rectangles at a cost of 2 beacons.
		\end{lemma}
				
		\begin{figure}[htbp] 
			\begin{center}
				\includegraphics*[scale=1.25]{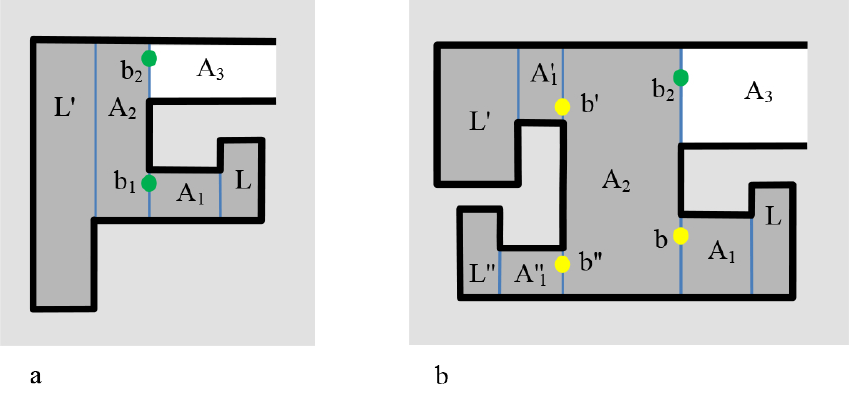} 
			\end{center}
			\caption{ (a) $A_2$ has a Type IV neighbor. 
					  (b) All children of $A_2$ are Type I.}
			\label{fig:typefour}
		\end{figure}
		
		\begin{proof}
			If $A_2$ has a Type IV subtree, then that subtree is on the left of $A_2$,
			because it is tall and $A_3$ is on the right.  Also because it is tall, there
			is no other Type II or Type IV subtree present.
			$A_2$ is some deepest leaf's grandparent, so $A_2$ must have at least one
			grandchild. Since the Type IV subtree is simply a child of $A_2$, the
			lower right neighbor of $A_2$ must be of Type I.
			This situation is illustrated in Figure \ref{fig:typefour}a.
			
			We place beacons $b_1$, on the vertical between
			$A_1$ and $A_2$, and $b_2$, on the vertical between $A_2$ and $A_3$.
			Beacon $b_1$ covers $A_1$ and $L$ (the Type I subtree)
			and beacon $b_2$ covers $L'$ (the Type IV subtree) and $A_2$.
			The beacons see one another, thus are strongly connected in the attraction
			graph, and $b_2$ is in the polygon $P_{k+1}$ remaining after the reduction.
			By Lemma \ref{lem:routingset}, then, the
			current lemma follows.
		\end{proof}
		
		We have now shown how to reduce the polygon whenever $A_2$ has a tall subtree.
		It remains for us to examine the cases where all of $A_2$'s subtrees are
		short.
		
		\begin{lemma}\label{lem:all-typeI}
		If $A_2$'s subtrees are all Type I, then $P_k$ can be reduced by 3 
			rectangles at a cost of 2 beacons, 5 rectangles at a cost of 3 beacons,
			or 7 rectangles at a cost of 4 beacons.
		\end{lemma}
		\begin{proof}
			We place one beacon $b_2$ on the vertical between $A_2$ and $A_3$,
			and one beacon for each subtree of $A_2$, on the vertical between
			$A_2$ and its subtree.  Figure \ref{fig:typefour}b shows the situation
			when $A_2$ has three subtrees. There are 7 rectangles removed
			and 4 beacons placed.
			
			When $A_2$ has two or one subtree, the situation will be as in the figure,
			but with one or two of the subtrees, and the corresponding beacons, removed.
			Also, with one or two subtrees removed, there is a possibility that $A_3$ is
			a tall neighbor of $A_2$; this is of no concern as we still place $b_2$ on the
			vertical between $A_2$ and $A_3$.
			
			So if $A_2$ has two subtrees, then there are 5
			rectangles removed and 3 beacons placed.  If $A_2$ has one subtree, then there are 3 rectangles
			removed and 2 beacons placed.
			
			All beacons are in $A_2$ and therefore cover $A_2$ and see one another.  This
			means they are strongly connected.  The beacon $b$ (or $b'$ or $b''$)
			corresponding to each subtree covers the rectangles $A_1$ and $L$ of that
			subtree.
			The beacon $b_2$ is in the polygon $P_{k+1}$ remaining after the reduction.
		\end{proof}
		
		We now need to consider only cases where there is at least one Type III
		subtree present.  Since $A_2$ has a grandchild, there must also be a Type I
		subtree.  We consider the alternatives for the third subtree of $A_2$:  it is
		either absent, Type I, or Type III.
		
		\begin{lemma}
			If $A_2$ has two Type I subtrees, and one Type III subtree, then $P_k$
			can be reduced by 6 rectangles at a cost of 4 beacons.
		\end{lemma}
		\begin{proof}
			The situation is as depicted in Figure \ref{fig:typefour}b, except that
			one of the leaf rectangles $L$, $L'$, or $L''$ is missing.
			This is handled in the same manner as Lemma \ref{lem:all-typeI},
			placing a beacon on the vertical between $A_2$ and each of its neighbors.
		\end{proof}
		
		\begin{lemma}\label{lem:onethree}
			If $A_2$ has exactly one Type I subtree, and exactly one Type III subtree,
			then $P_k$ can be reduced by 3 or 4 rectangles at a cost of 2 beacons.
		\end{lemma}
		\begin{figure}[htbp] 
			\begin{center}
				\includegraphics*[scale=1.25]{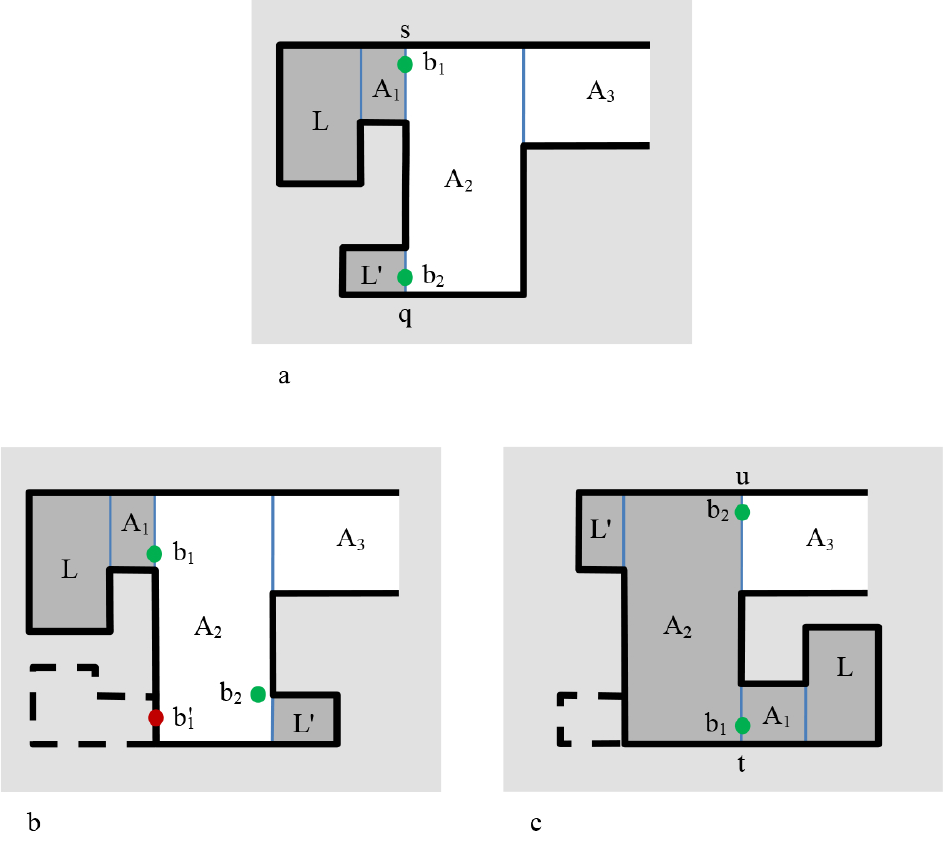} 
			\end{center}
			\caption{ $A_2$ has one Type I subtree and one Type III subtree.
					  (a) both subtrees are on the left of $A_2$.
					  (b) the Type III subtree is on the right.
					  (c) the Type I subtreee is on the right, and the Type III subtree is on
					  the left.
					}
			\label{fig:typeonethree}
		\end{figure}
		\begin{proof}
			Let the Type I subtree have rectangles $L$ and $A_1$, 
			and the Type III subtree have rectangle $L'$.
			We consider the three possibilities:
			both subtrees are on the left of $A_2$,
			the Type III subtree is on the lower right of $A_2$,
			and the Type I subtree is on the lower right of $A_2$.
			
			In the first case, we place a beacon $b_1$ at $s - \varepsilon\hat{y}$,
			and a beacon $b_2$ at $q + \varepsilon\hat{y}$,
			where $q$ and $s$ are the lower-left and upper-left corners of $A_2$,
			respectively.
			We remove $L$, $A_1$, and $L'$.  This is shown in Figure
			\ref{fig:typeonethree}a when $A_1$ is an upper-left neighbor of $A_2$; the
			case when $A_1$ is a lower-left neighbor is similar and not shown.  Here,
			$A_3$ may be a short (solo) neighbor of $A_2$, as pictured, or it may be a
			tall neighbor.
			
			In the second case, we place a beacon $b_1$ on the vertical between 
			$A_1$ and $A_2$, and a beacon $b_2$ at $r-\varepsilon\hat{x}$,
			where $r$ is the reflex vertex shared by $A_2$ and $L'$.
			This is repair position for any paths that get trapped in this reduction.
			We remove $L$, $A_1$, and $L'$. 
			This is shown in Figure \ref{fig:typeonethree}b for $A_1$ in the upper left;
			$A_1$ in the lower left is as suggested by the dashed boundary and red
			beacon placement.
			
			In the third case, we place beacons $b_1$ at $t + \varepsilon\hat{y}$ and
			$b_2$ at $u - \varepsilon\hat{y}$, where $t$ and $u$ are the lower right and
			upper right corners of $A_2$.  We remove the four rectangles $A_2$,
			$A_1$, $L$, and $L'$.
			If $L'$ is an upper-left neighbor of $A_2$, then the situation is as depicted
			in Figure \ref{fig:typeonethree}, and $b_2$ covers $L'$.
			If $L'$ is instead a lower-left neighbor of $A_2$, then the situation is
			suggested with the dashed boundary in the figure, and $b_1$ covers $L'$.
						
			In all cases in this lemma, $b_1$ and $b_2$ are both in $A_2$, so they see
			one another.
			Also, $b_2$ is always in $P_{k+1}$.
			The beacon $b_1$ always covers $A_2$, $A_1$, and $L$.
			In all but the last case, $b_2$ covers $L'$; in the last case,
			$b_1$ or $b_2$ covers $L'$.
			In the one case that trapped paths could occur, $b_2$ was in repair position.
			By Lemma \ref{lem:routingset}, the current lemma follows.
		\end{proof}

		\begin{lemma}\label{lem:onethreethree}
			If $A_2$ has one Type I subtree, and two Type III subtrees,
			then $P_k$ can be reduced by 3 or 5 rectangles at a cost of 2 beacons.
		\end{lemma}
		\begin{figure}[htbp] 
			\begin{center}
				\includegraphics*[scale=1.25]{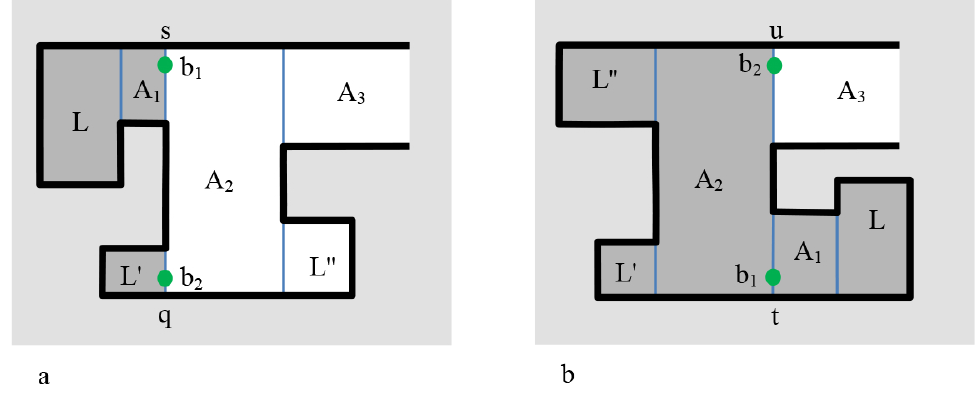} 
			\end{center}
			\caption{ $A_2$ has one Type I subtree and two Type III subtrees.
					  (a) The Type I subtree is on the left.
					  (b) The Type I subtree is on the right.
					}
			\label{fig:onethreethree}
		\end{figure}  
		
		\begin{proof}
			If the Type I subtree is on the upper left, then the situation is
			as in Figure \ref{fig:onethreethree}a, we may
			apply the same reduction used in the first case of Lemma \ref{lem:onethree}.
			The case when the Type I subtree is on the lower left is similar.
			
			If the Type I subtree is on the right, then the situation is as shown
			in Figure \ref{fig:onethreethree}b.
			Here we apply the same reduction used in the last case of Lemma
			\ref{lem:onethree}, placing beacons $b_1$ at $t + \varepsilon\hat{y}$ and
			$b_2$ at $u - \varepsilon\hat{y}$, where $t$ and $u$ are the lower right and
			upper right corners of $A_2$.
			Here, $b_1$ covers all removed rectangles except the
			upper-left neighbor of $L''$, which $b_2$ covers.
			$b_1$ and $b_2$ see each other, and there is no possibility of trapped paths.
		\end{proof}
	
		We now summarize the last four sections.
		
		\begin{thm}\label{thm:casesummary}
			If $T_k$ has depth at least 3, then $P_k$ can be reduced by 2 rectangles at a
			cost of 1 beacon; 3, 4, or 5 rectangles at a cost of 2 beacons;  5 rectangles
			at a cost of 3 beacons; or 6 or 7 rectangles at a cost of 4 beacons.
			The reduction removes at most three layers from the dual tree.
		\end{thm}
		\begin{proof}
			Start at a deepest leaf $L$ of $T_k$, and label its parent and ancestors
			$A_1$, $A_2$, $A_3$, etc.  If there is a reduction from Sections
			\ref{sec:atwotall} to \ref{sec:atwopaired} at any child of $A_2$,
			we are done (the number of rectangles and beacons for the reduction is
			listed here in the theorem statement).  These reductions remove at most two
			layers from the dual tree.
			
			Otherwise, all children of $A_2$ are one of the four types in Figure
			\ref{fig:threeleveltypes}.
			If any of these subtrees are tall, then Lemma \ref{lem:typeII} or
			\ref{lem:typeIV} applies, and if they are all short, then one of Lemmas
			\ref{lem:all-typeI}--\ref{lem:onethreethree} applies.
			Again, the number of rectangles and beacons in the reduction is listed here;
			these reductions remove at most three layers from the dual tree.
		\end{proof}
		
		\begin{cor}\label{cor:casescondensed}
			If $T_k$ has depth at least 3, then $P_k$ can be reduced by some $s$
			rectangles at a cost of $b$ beacons, where $b \leq \floor{2s}{3}$.
			The reduction removes at most three layers from the dual tree.
		\end{cor}
		
	\subsection{Induction basis}
		The basis for our induction is when $T_k$ has depth 2 or smaller.
				The basis cases are when there are only one or two levels in the dual tree.
				
		If it has depth 0,
		the tree is simply a node, and 
		the polygon is a rectangle.
		
		If it has depth 1,
		since we rooted it at a leaf, then the tree has only two nodes, and
		the polygon is a 6-vertex ``L'' shape.
		In both of these cases, every point in the polygon attracts every other point
		in the polygon (see Lemma \ref{lem:solocovered}).  Thus, there are no intermediate
		beacons required and the smallest beacon routing set is of size 0.
		\begin{figure}[htbp] 
			\begin{center}
				\includegraphics*[scale=1.25]{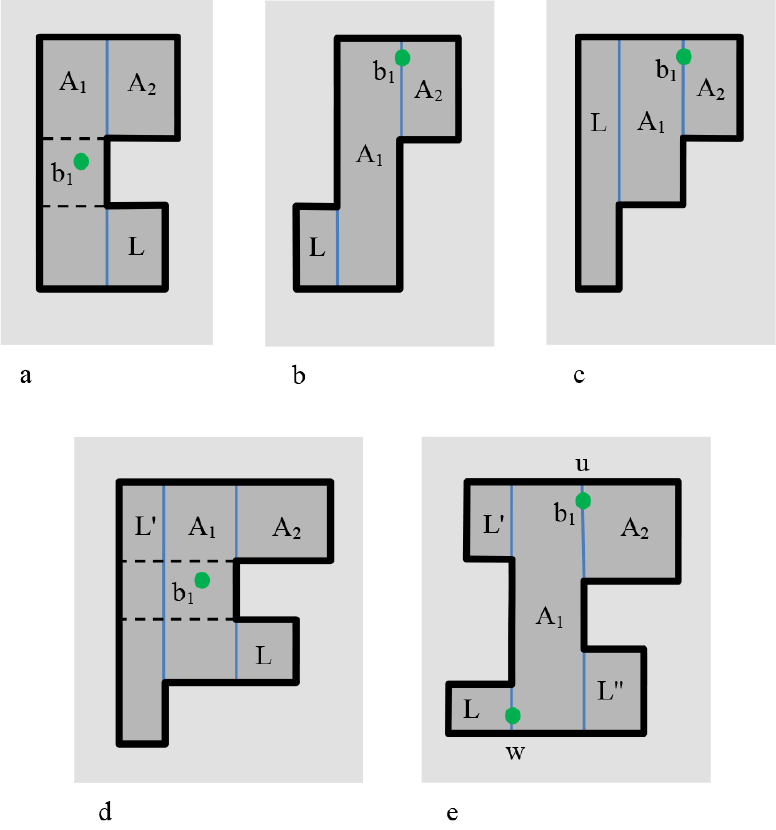} 
			\end{center}
			\caption{ The dual tree has depth 2. 
					(a) $A_1$ has one child on the right.
					(b) $A_1$ has one child, on the left and short.
					(c) $A_1$ has one child, on the left and tall.
					(d) $A_1$ has two or three children.
					}
			\label{fig:basis}
		\end{figure}  
		The depth-2 situation is a little more involved.  $A_2$'s only child is $A_1$,
		but $A_1$ has one to three children.  This gives a total of 3 to 5
		rectangles, or $n = 8$ to $12$.  As above, we assume that $A_2$ is an upper
		right neighbor of $A_1$.
		
		If $A_1$ has one child, then there are three rectangles and $n = 8$.
		If the neighbors of $A_1$ are both right neighbors, then the situation is
		as depicted in Figure \ref{fig:basis}a, and we
		cover the polygon with one beacon in the modified right center of $A_1$, by
		Lemma \ref{lem:pairedcovered}. 
		If there is one left neighbor $L$ and one right neighbor $A_2$, then we cover
		them with a beacon $b_1$ on the vertical between $A_1$ and $A_2$.
		If the left neighbor is short, as shown in Figure \ref{fig:basis}b for a
		lower-left neighbor, the beacon $b_1$ covers $L$ by Lemma
		\ref{lem:sololeafcovered}.
		If instead the left neighbor is tall, as shown in Figure \ref{fig:basis}c for
		an upper-left neighbor, the beacon $b_1$ covers $L$ by Lemma
		\ref{lem:tallleafcovered}.
		(The cases of a short upper-left neighbor and of a tall lower-left neighbor
		are similar.)  Since $\floor{8-4}{3} = 1$, we have covered the
		polygon with a correct number of beacons.
		
		If $A_1$ has two children, and one of its neighbors is tall, then
		the situation is as depicted in Figure \ref{fig:basis}d (or symmetric to it).
		Here we cover the polygon with one beacon placed in the modified
		right center of $A_1$.
		
		Otherwise, if $A_1$ has two or three children, and all of its neighbors are
		short, then there are three or four rectangles, giving $n = 10$ or $12$.
		$\floor{10-4}{3} = \floor{12-4}{3} = 2$, so we have two beacons with which to
		cover the polygon.
		The situation is as depicted in Figure \ref{fig:basis}e, although one of
		$L$, $L'$, or $L''$ may be missing.
		$A_1$ must have at least one lower neighbor, say $L$,
		on either the left or the right.
		We place beacon $b_1$ at $u - \varepsilon\hat{y}$ and 
		$b_2$ at $w + \varepsilon\hat{y}$,
		where $u$ is the upper-right corner of $A_1$, and 
		$w$ is the lower shared corner of $L$ and $A_1$.
		The top beacon $b_1$ covers $A_1$ and $A_1$'s top neighbors, and the bottom
		beacon $b_2$ covers $A_1$'s bottom neighbors.  The beacons are visible to each
		other, so they form a routing beacon set.
		
		We have thus shown the following:
		
		\begin{lemma}\label{lem:basis}
		If the dual tree has depth two or smaller, then the polygon has a beacon
		routing set of $\floor{n-4}{3}$ beacons.
		\end{lemma}
		
	\subsection{Lower bound}
	
		\begin{thm}
			Any orthogonal polygon of $n$ vertices has a local beacon routing set of at
			most $\floor{n-4}{3}$ beacons.
		\end{thm}
		\newcommand{\nrect}{\ensuremath{r}}
		
		\begin{proof}
			Let \nrect\  be the number of rectangles in the vertical decomposition of the
			polygon.  Since $n = 2\nrect + 2,$ the floor in the theorem is equivalent to
			$\floor{(2\nrect + 2)-4}{3} = \floor{2\nrect - 2}{3}.$  We proceed to prove
			that there is a beacon set no larger than this, by induction on \nrect.
			First, we root the dual tree at a leaf.
			
			We stop the induction when the dual tree has depth two or smaller, measured
			from this root.
			Lemma \ref{lem:basis} establishes these polygons as satisfying the theorem.

			For our inductive step, the depth of the dual tree is at least 3.
			Thus Theorem \ref{thm:casesummary} applies, and gives us a reduction of
			$s$ rectangles for $b$ beacons, where $b < \floor{2s}{3}$.
			
			We reduce $P$ by $s$ rectangles to construct a
			$P'$ with $\nrect' = \nrect-s$ rectangles.
			We know that $\nrect' > 0$ since the dual tree has depth at least three
			(i.e., at least four levels)
			and the reductions remove at most three levels from that.
			So by induction $P'$ has a local beacon routing set of at most 
			$\floor{2\nrect'- 2}{3}= \floor{2(\nrect-s)-2}{3} = 
			\floor{2\nrect- 2s - 2}{3}$ beacons.
			To construct the beacon set for $P$, we add $b$ beacons to that, 
			and so we have at most
			$\floor{2\nrect- 2s - 2}{3} + b \leq 
			 \floor{2\nrect- 2s - 2}{3} + \floor{2s}{3} \leq
			 \floor{2\nrect - 2}{3}$ beacons.
		\end{proof}

\section{Lower bound} 

\newcommand{\mXin}[1]{\ensuremath{m_{#1}^{\mbox{\small in}}}}
\newcommand{\mXout}[1]{\ensuremath{m_{#1}^{\mbox{\small out}}}} 
\newcommand{\mkin}{\mXin{k}}		
\newcommand{\mkout}{\mXout{k}}
\newcommand{\mimin}{\mXin{3i-2}}
\newcommand{\mimout}{\mXout{3i-2}}
\newcommand{\mipin}{\mXin{3i+1}}
\newcommand{\mipout}{\mXout{3i+1}}
\newcommand{\Hp}[1]{\ensuremath{H^{+}_{#1}}}
\newcommand{\Hm}[1]{\ensuremath{H^{-}_{#1}}}

	In this section we exhibit an infinite class of orthogonal polygons
	that require
	\floor{(n-4)}{3}\ beacons to route between any pair of points.
	The examples are geometrically simple, being orthogonal spiral
	polygons with a corridor width of $1$.

	Our polygons will spiral outwards clockwise as one moves through the reflex
	chain when walking counterclockwise around the polygon (i.e. left hand on
	interior).
	Call the reflex vertices of the polygon $r_1, r_2, \ldots r_{(n-2)/2}$ in
	this counterclockwise order, and let $r_0$ and $r_{n/2}$ denote the convex
	vertices adjacent to $r_1$ and $r_{(n-2)/2}$, respectively.
	Let $c_k$ be the convex vertex just outside of (and closest to) $r_k$
	(refer to Figure \ref{fig:spiralDefinitions}).
	Let $e_k$ be the edge from $r_k$ to $r_k+1$, and $l_k$ be the length of $e_k$.

	Now let $C_k$ be the ``corner'' $1$ by $1$ square in $P$ with vertices $r_k$ and
	$c_k$, and $H_k$ be the ``hallway'' rectangle (with dimensions $1$ by $l_k$)
	between $C_{k-1}$ and $C_{k}$.

	If \mkin\  is the midpoint of $r_{k-1}$ and $r_k$,
	and \mkout\  is the midpoint of $c_{k-1}$ and $c_k$,
	we can partition the ``hallway'' $H_k$ into two halves \Hp{k}\ and \Hm{k}\ 
	by splitting with its bisector \mkin\mkout.
	Let \Hp{k}\ be the half adjoining $C_k$, and
	let that half (and not \Hm{k}) contain the points on the segment
	\mkin\mkout.

\begin{figure}[htb]
	\begin{center}
		\includegraphics{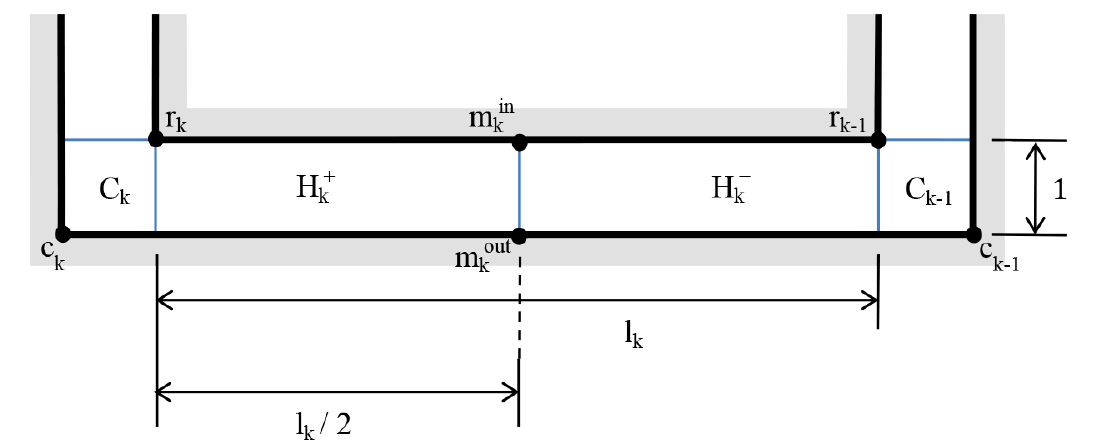}
	\end{center}
	\caption{ Notation for an orthogonal spiral. }
	\label{fig:spiralDefinitions}
\end{figure}

We will construct polygons for $n = 6r + 4$ for some $r$; these polygons
are specified simply by giving the lengths $l_1, l_2, \ldots l_{3r+1}$ 
of the $3r+1$ ``hallway'' rectangles.
Provided we have $l_j > l_{j-2} + 2$ for all $3 \leq j \leq 3r$, the polygon
will spiral outward and not self-intersect.

We specify $r$ sections $S_1, S_2, \ldots S_r$ of the polygon, by letting
$S_i$ be the union of $\Hp{3i-2}, C_{3i-2}, H_{3i-1}, C_{3i-1}, H_{3i},
C_{3i},$ and $\Hm{3i+1}$ (see Figure \ref{fig:spiralSection}).
Note that no point of $P$ is contained in more than one section, and there
are points at either end of the spiral (in \Hm{1} and \Hp{3r+1})
that are in no section.

Now consider a set of beacons $B$ that can route in such a polygon $P$.
We claim that $|B| >= 2r$.
If this were not the case, then by the pigeonhole principle some section $S_i$
would contain less than two beacons.

\begin{figure}[htb] 
	\begin{center}
		\includegraphics{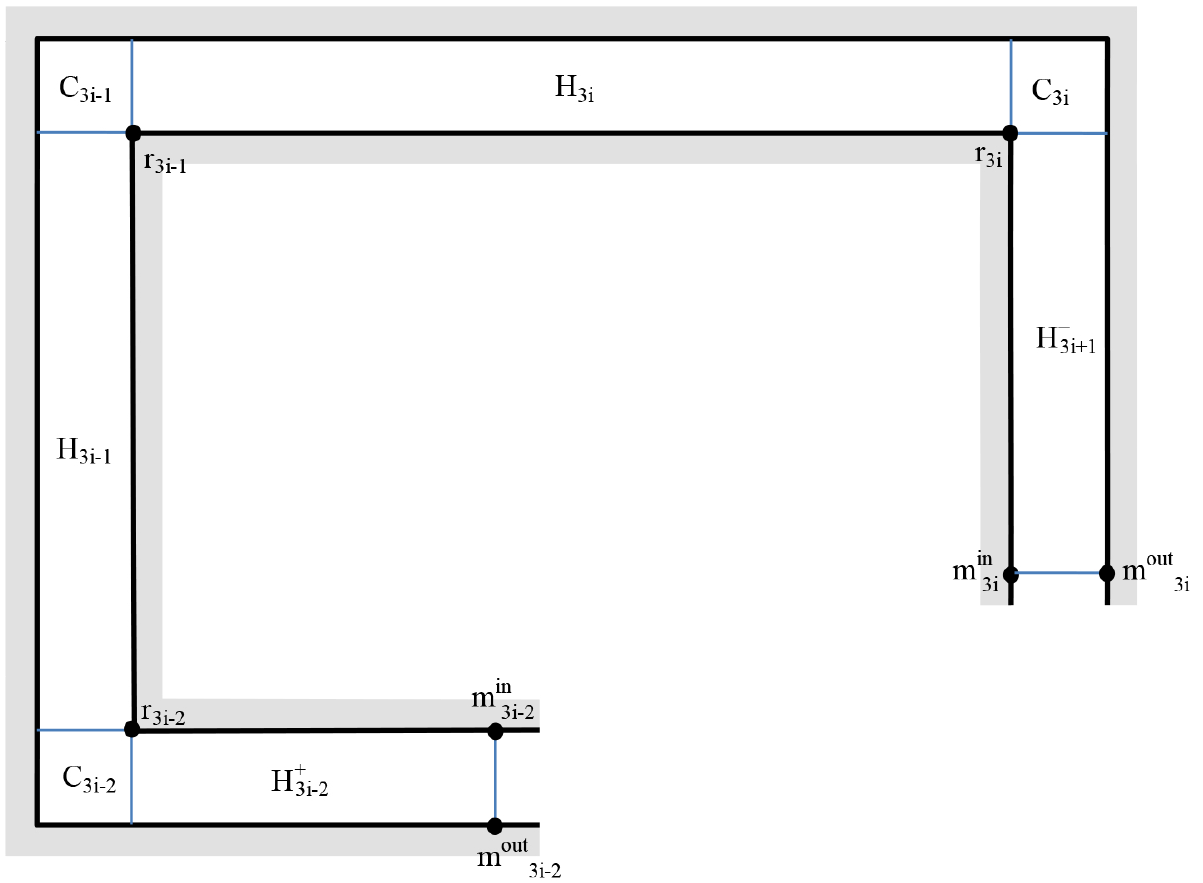} 
	\end{center}
	\caption{ A section of an orthogonal spiral. }
	\label{fig:spiralSection}
\end{figure}

If $S_i$ is removed from $P$, then there are two nonempty subpolygons left:
the part \emph{before} $S_i$, which contains at least \Hm{3i-2},
and the part \emph{after} $S_i$, which contains at least \Hp{3i+1}.
Since $B$ is a routing set of beacons, one must be able to route between
a point in the part before $S_i$ to a point in the part after $S_i$ using only
the beacons of $B$.
In order for a robot to get from a point before $S_i$ to a point after,
it must at some point pass from \Hp{3i-2}\ to $C_{3i-2} \cup H_{3i-1}$ at some
point on the (closed) vertical between $C_{3i-2}$ and \Hp{3i-2} (refer to
Figure \ref{fig:spiralSection})
 For a beacon to cause this to happen, the beacon must be on or left of the
 line $r_{3i-1}r_{3i-2}$ in one of the three local regions $C_{3i-1}, H_{3i-1},$
 and $C_{3i-2}$. (If it is not in one of these three local regions, the robot 
 will become stuck without reaching the beacon.)

To summarize, some $S_i$ has fewer than two beacons, but to route from a point
before $S_i$ to a point after $S_i$, there must be a beacon in $C_{3i-1}, H_{3i-1},$
 or $C_{3i-2}$.  Thus, the no-beacon option is eliminated, and this $S_i$ has
 one beacon.  
 
Now consider routing from some point after $S_i$ to some point before $S_i$.
An argument symmetric to that above shows that $S_i$ must have a beacon in
$C_{3i-1}, H_{3i},$ or $C_{3i}$.

Thus, the single beacon $b$ in $S_i$ lies in $C_{3i-1}$.
Consider again routing from some point before $S_i$ to some point after.
After activating $b$ and attracting the robot there, another beacon must
activate and attract the robot along the next stage of its routing.
Since there are no other beacons in $S_i$, and since we can only use $b$
once, the next beacon must be somewhere after $S_i$.
For a beacon to successfully attract a robot from $C_{3i-1}$ to somewhere
after $S_i$, the beacon must be in either \Hp{3i+1}\ or $C_{3i+1}$. 

\newcommand{\bafter}{\ensuremath{b^{\mbox{\small after}}}}
\newcommand{\ebad}{\ensuremath{r_{3i-2}r_{3i-1}}}
\newcommand{\egood}{\ensuremath{r_{3i-1}r_{3i}}}
Since the hallways $H_{3i}$ and $H_{3i+1}$ are (considerably) longer than they
are wide, a robot in $C_{3i-1}$ attracted towards a beacon \bafter\ in \Hp{3i+1}
or $C_{3i+1}$ will either hit \ebad\  (as shown in red in Figure
\ref{fig:attractCorner}), hit \egood\  (as shown in green in
the figure), or hit the reflex vertex $r_{3i-1}$ itself.
If the robot hits \ebad\  it will eventually get stuck,
but if it hits \egood\  it will continue along this wall,
eventually reaching \bafter.

\begin{figure}[htbp] 
	\begin{center}
		\includegraphics{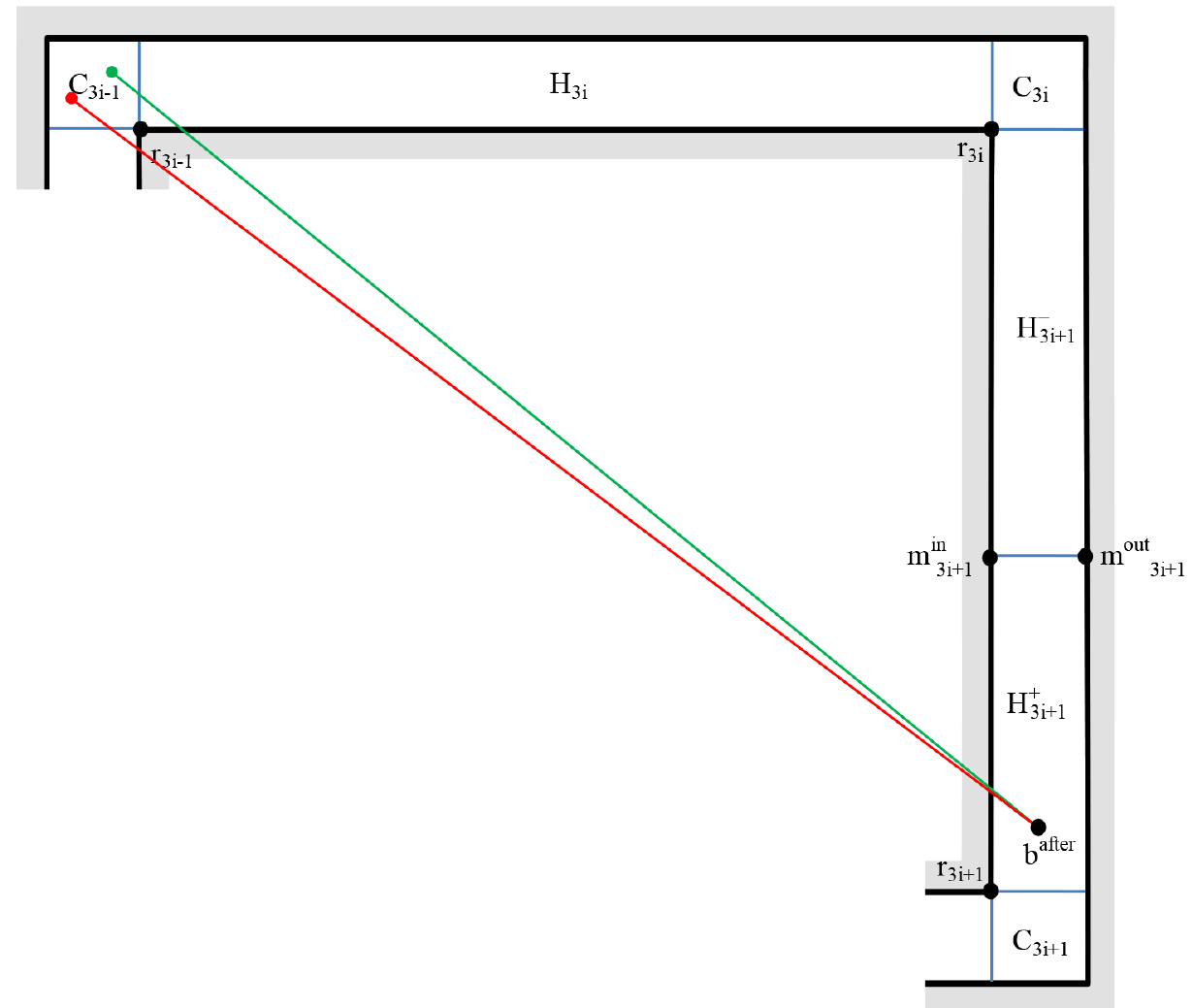} 
	\end{center}
	\caption{ 
		A robot in $C_{3i-1}$ attracted to a beacon in \Hp{3i+1} or 
		$C_{3i+1}$.
	}\label{fig:attractCorner}
\end{figure}

If the robot is below (with reference to Figure \ref{fig:attractCorner})
the line \bafter$r_{3i-1}$, then it will hit \ebad\ and get stuck.
Thus, the robot (and hence the single beacon in $S_i$) must be located
on or above \bafter$r_{3i-1}$ in $C_{3i-1}$.
Since this need be true only for a single beacon \bafter\ in \Hp{3i+1}\ or
$C_{3i+1}$, we can assume the most permissive case of $\bafter = \mipout$,
and derive that the beacon in $S_i$ must be located on or above
\mipout$r_{3i-1}$ in $C_{3i-1}$ (the green-striped region in Figure
\ref{fig:aboveLine}).
That is, any robot below this line would be attracted into \ebad\ by any beacon
in \Hp{3i+1}\ or $C_{3i+1}$.

\begin{figure}[htbp] 
	\begin{center}
		\includegraphics{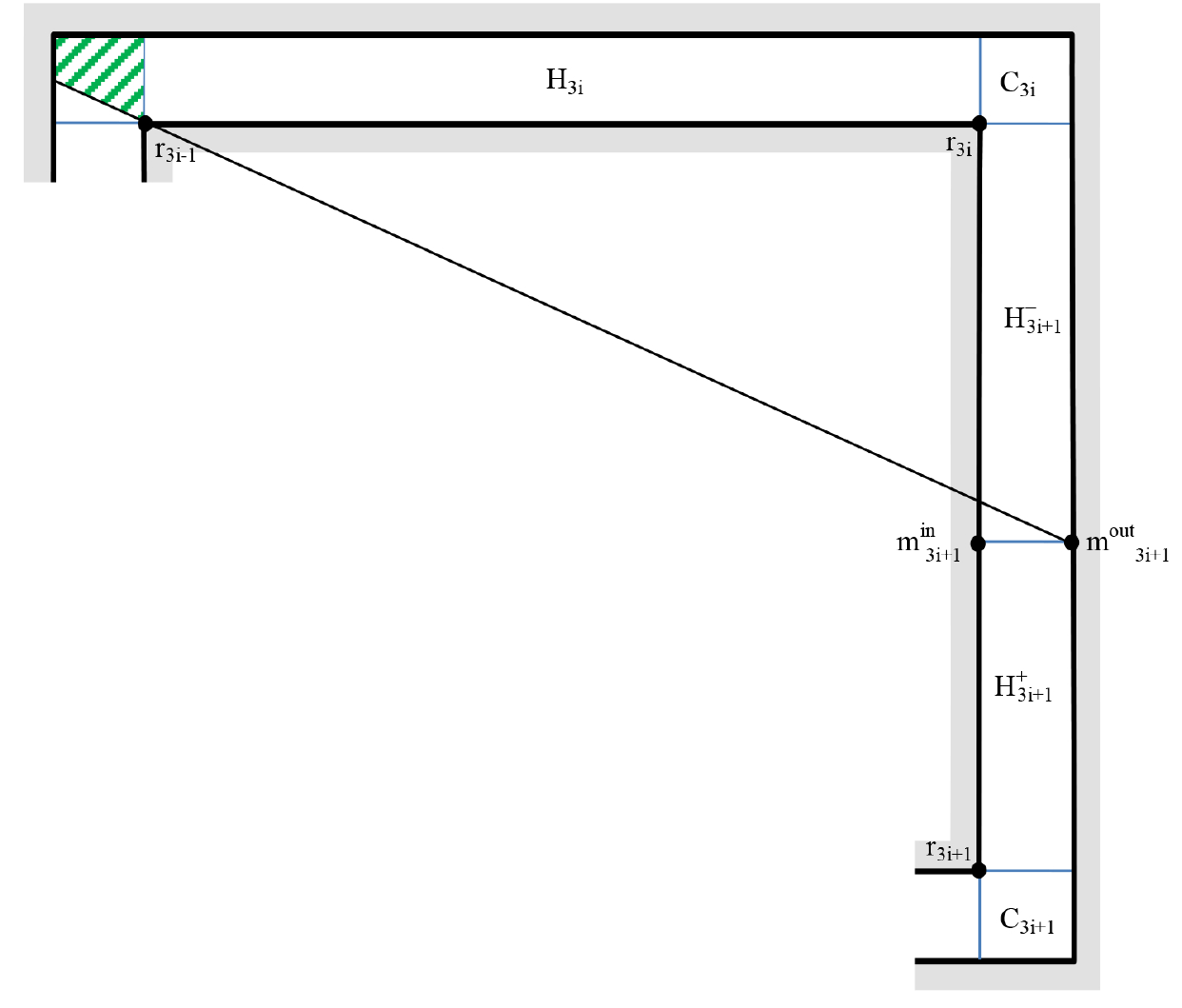} 
	\end{center}
	\caption{The region in which the beacon in $S_i$ must lie.}
	\label{fig:aboveLine}
\end{figure}

By a symmetric argument, considering a routing from some point before $S_i$ to
some point after $S_i$, we get that the beacon in $S_i$ must be located on or
below \mimout$r_{3i-1}$.  The effect of this constraint combined with the
previous one is illustrated in Figure \ref{fig:betweenLines}.
However, Figure \ref{fig:betweenLines} is not the only geometric situation
possible:
if the ratio of $l_{3i+1}/2$ to $l_{3i}$ is greater than the ratio of
$l_{3i-1}+1$ to $l_{3i-2}/2$, then there are no points of $C_{3i-1}$ other than
$r_{3i-1}$ that satisfy both constraints; this is illustrated in Figure
\ref{fig:crossingLines}.

We rewrite this inequality on the ratios of the corridor lengths as
	\[
		\frac{l_{3i+1}}{2l_{3i}} > \frac{2(l_{3i-1}+1)}{l_{3i-2}}
	\]
and multiply both sides by $2l_{3i}$ to obtain
	\[
		l_{3i+1} > \frac{4 l_{3i} (l_{3i-1}+1)}{l_{3i-2}}
	\]
which we shall refer to as the \emph{length inequality}.

Thus far we have shown that, if $S_i$ contains less than two beacons, and
it satisfies the length inequality, then $S_i$ contains exactly one beacon
at $r_{3i-1}$.  
In this situation, consider a before-to-after-$S_i$ routing of a robot, and an
after-to-before-$S_i$ routing.
The next beacon on either of these routings
(being in either $\Hp{3i+1} \cup C_{3i+1}$ 
 or $\Hm{ei-2} \cup C_{3i-3}$)
would pull a robot at $r_{3i-1}$
locally towards the exterior of the polygon.
If the beacon-attraction model specifies either a fixed choice (along the
clockwise edge, or along the clockwise edge) or an arbitrary choice (one can't
tell ahead of time which of the edges the robot will choose to move along)
for a robot pulled towards the exterior of a reflex vertex, then in at least
some instances on one of the before-to-after and after-to-before routings, the
robot goes along the wrong edge and gets stuck.
Thus, even $r_{3i-1}$ is not a valid choice for a single beacon in $S_i$ in a
valid routing set of beacons $B$ when the length inequality holds.

Given the length inequality, we have now eliminated all possibilites for $S_i$
to contain fewer than two beacons, so $S_i$ contains at least two beacons, and
the polygon therefore contains at least $2r$ beacons; since $n = 6r + 4$, we can
rewrite the number of beacons as at least $(n-4)/3$.

We now show how to choose lengths $l_1, l_2, \ldots, l_{r+1}$ so that the length
inequality holds for each $1 \leq i \leq r$, and so that the polygon spirals
outwards without self-intersection.

We will enforce the length inequality for each $l_k$ (where $k > 3$) as the
left-hand side, rather than simply for those $k$ that are equivalent to $1$ modulo $3$:
\[
	l_{k} > \frac{4 l_{k-1} (l_{k-2}+1)}{l_{k-3}}
\]
And we will replace this with the stronger requirement
\[
	l_{k} \geq \frac{8 l_{k-1} l_{k-2}}{l_{k-3}}
\]
by requiring that every $l_{k-2} > 1$.

By letting $m_k = \log l_k$, we get the recurrence
\[
	m_k \geq 3 + m_{k-1} + m_{k-2} - m_{k-3}
\]
which has the solution
\[
	m_k = k^2,
\]
as one can verify by substitution.
(If we change the inequality to an equality and solve the recurrence exactly,
we still get a function in $\Theta(k^2)$.)
So if we choose $l_k = 2^{m_k} = 2^{k^2}$, then the length inequality is
everywhere satisfied.
It is also simple to verify our requirement $l_{k-2} > 1$ is always satsified.

To ensure that the polygon spirals outward without
self-intersection, we only require that $l_k > 2 + l_{k-2}$ for all $3 \leq k
\leq r$.  Again, with our choice of $l_k = 2^{k^2},$ this is easily verified.

In sum, the $(6k + 4)$-vertex rectangular spiral with hallway lengths $l_k =
2^{k^2}$ requires at least $2k = (n-4)/3$ beacons in a beacon set for routing.

\begin{figure}[htb] 
	\begin{center}
		\includegraphics{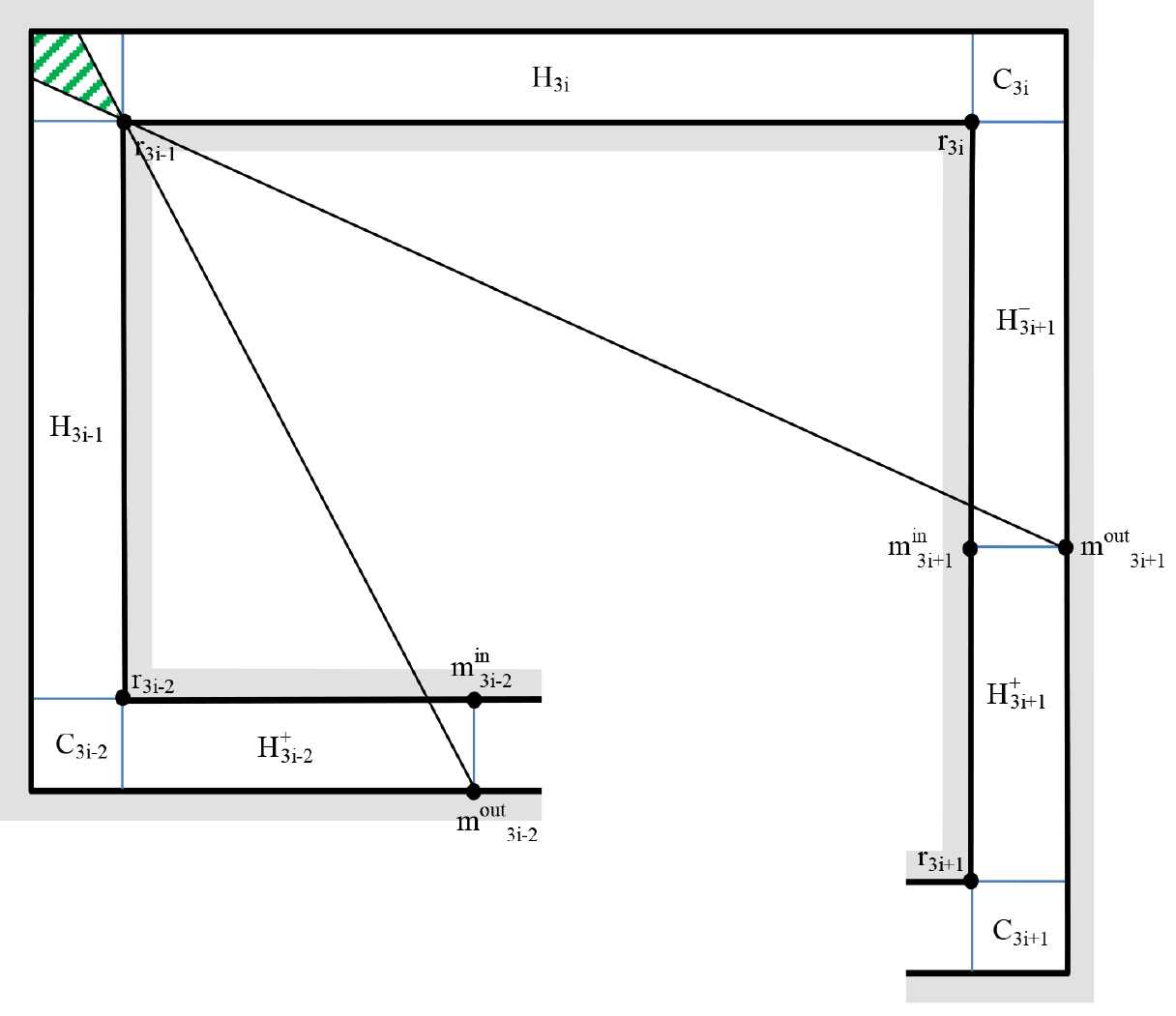} 
	\end{center}
	\caption{
		Adding the symmetric constraint.  The beacon
		in $S_i$ must lie in the shaded area.
	}
	\label{fig:betweenLines}
\end{figure}  

\begin{figure}[htb] 
	\begin{center}
		\includegraphics{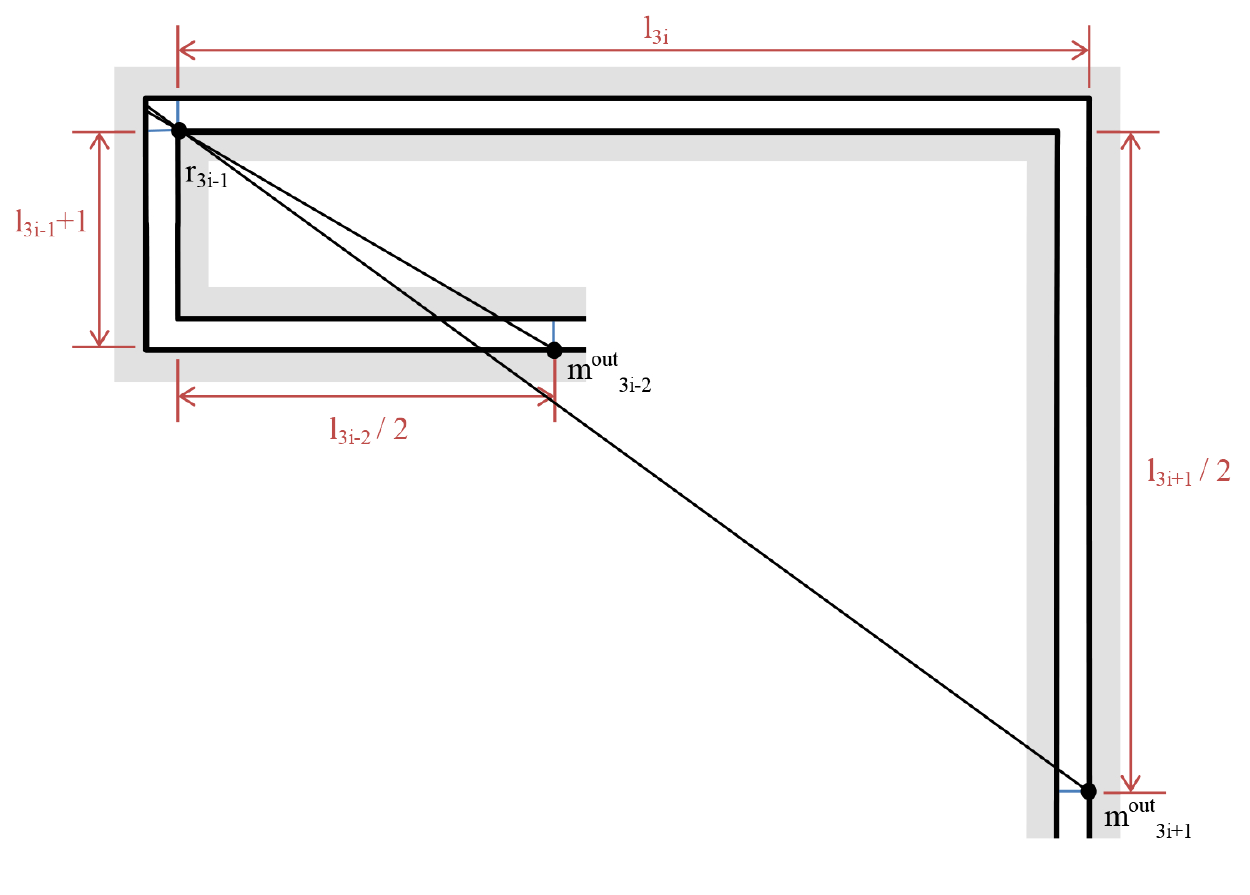} 
	\end{center}
	\caption{
		If the lengths satisfy the length inequality, the allowable region is
		only the point $r_{3i-1}$.
	}\label{fig:crossingLines}
\end{figure}

\bibliographystyle{abbrv}

\end{document}